\RequirePackage{fix-cm}
\documentclass[twocolumn]{svjour3}          
\smartqed  

\usepackage{graphicx,epstopdf,url,amsthm,color}
\usepackage[cmex10]{amsmath}

\usepackage{balance}  
\usepackage{subfigure}
\usepackage{graphicx,epstopdf,url,color}
\usepackage{balance}
\usepackage{multirow}
\usepackage{xcolor}
\usepackage{algpseudocode}
\usepackage{amssymb}
\usepackage[lined,boxed,commentsnumbered,ruled,vlined]{algorithm2e}
\usepackage{float}
\usepackage{array}
\usepackage{multicol}
\usepackage{amsmath}
\usepackage[misc]{ifsym}
\usepackage{xcolor}
\newcommand\mycolor[1]{\textcolor{black}{#1}} 
\newfloat{algorithm2e}{t}{lop}
\newcolumntype{C}[1]{>{\centering\let\newline\\\arraybackslash\hspace{0pt}}m{#1}}
\journalname{}
\begin{document}
	\title{Multi-Constraint Shortest Path using Forest Hop Labeling}
	\author{Ziyi Liu    \and
		Lei Li			\and
		Mengxuan Zhang	\and
		Wen Hua			\and
		Xiaofang Zhou
	}    
    
   \institute{Ziyi Liu$^*$ \email{ziyi.liu@uq.edu.au} 
   			\\Lei Li$^*$(\Letter) \email{l.li3@uq.edu.au} 
   			\\Mengxuan Zhang$^*$(\Letter)  \email{mengxuan.zhang@uq.edu.au} 
			\\Wen Hua$^*$ \email{w.hua@uq.edu.au} 
			\\Xiaofang Zhou$^{\dag}$ \email{zxf@cse.ust.hk}\\
		\at $^*$School of Information Technology and Electrical Engineering, University of Queensland, QLD 4072, Australia 
		\at $^{\dag}$Department of Computer Science and Engineering, The Hong Kong University of Science and Technology, Clear Water Bay, Kowloon, Hong Kong\\
	}
    \date{Received: date / Accepted: date}
	\maketitle
	
	\begin{abstract}
    The \textit{Multi-Constraint Shortest Path (MCSP)} problem aims to find the shortest path between two nodes in a network subject to a given constraint set. It is typically processed as a \textit{skyline path} problem. However, the number of intermediate skyline paths becomes larger as the network size increases and the constraint number grows, which brings about the dramatical growth of computational cost and further makes the existing index-based methods hardly capable of obtaining the complete exact results.
    In this paper, we propose a novel high-dimensional skyline path concatenation method to avoid the expensive skyline path search,
    which then supports the efficient construction of hop labeling index for \textit{MCSP} queries.
    Specifically, a set of insightful observations and techniques are proposed to improve the efficiency of concatenating two skyline path set, a \textit{n-Cube} technique is designed to prune the concatenation space among multiple  hops, and a \textit{constraint pruning} method is used to avoid the unnecessary computation.
    Furthermore, to scale up to larger networks, we propose a novel \textit{forest hop labeling} which enables the parallel label construction from different network partitions.
    Our approach is the first method that can achieve both accuracy and efficiency for \textit{MCSP} query answering. Extensive experiments on real-life road networks demonstrate the superiority of our method over the state-of-the-art solutions.

	\end{abstract}

\section{Introduction}
\label{sec:Introduction}
Route planning is an important application in our daily life, and various path problems have been identified and studied in the past decades.
In the real-life route planning, there are always multiple criteria depending on users' preferences. For example, a user may have a limited budget to pay the toll charge of highways, bridges, tunnels, or congestion. Some mega-cities require drivers to reduce the number of big turns (e.g., left turns in the right driving case\footnote{\mycolor{https://www.foxnews.com/auto/new-york-city-to-google-reduce-the-number-of-left-turns-in-maps-navigation-directions}} 
as they have a higher chance to cause accidents). There could be other side-criteria such as the minimum height of tunnels, the maximum capacity of roads, the maximum gradient of slopes, the total elevation increase, the length of tourist drives, the number of transportation changes, etc.
However, most of the existing algorithms only optimize one objective rather than the other side-criteria, such as minimum distance \cite{dijkstra1959note,li2020fast}, travel time of driving \cite{li2019Time,li2020fastest} or public transportation \cite{wang2015efficient,wu2014path}, fuel consumption \cite{li2017minimal,li2018go}, battery usage \cite{adler2014online}, etc. In fact, it is these criteria and their combination that endow different routes with diversified meaning and satisfy users' flexible needs. In other words, this \textit{multi-criteria} route planning is in more generalized manner, while the traditional \textit{shortest path} is only a special case. Moreover, it is an important research problem in the fields of both transportation \cite{jozefowiez2008multi,shi2017multi} and communication (Quality-of Service constraint routing) \cite{korkmaz2001multi,de1998multiple,tsaggouris2009multiobjective,van2003complexity}.


However, we could hardly find one result that is optimal in every criterion when there are more than one optimization goals. One way to implement the multi-criteria route planning is through \textit{skyline path} \cite{kriegel2010route,gong2019skyline,ouyang2018towards}, which provides a set of paths that cannot dominate each other in all criteria. However, skyline paths computation is very time-consuming with time complexity $O(c^{n-1}_{max}\times|V|\times(|V|\log |V|+|E|))$, where $|V|$ is the vertex number, $|E|$ is the edge number, $c_{max}$ is the largest criteria value, and $n$ is the number of criteria \cite{hansen1980bicriterion} ($c^{n-1}_{max}$ is the worst case skyline number). In addition, it is impractical to provide users with a large set of potential paths to let them choose from. Therefore, the \textit{Multi-Constraint Shortest Path (MCSP)}, as another way to consider the multiple criteria, is widely applied and studied. Specifically, it finds the best path based on one objective while requires other criteria satisfying some predefined constraints. For example, suppose we have three objectives such as distance $d$ and costs $c_1, c_2$, and we set the maximum constraints $C_1$ and $C_2$ on the corresponding costs. Then this \textit{MCSP} problem finds the shortest path $p$ whose cost $c_1(p)$ is no larger than $C_1$ and $c_2(p)$ is no larger than $C_2$. 

To the best of our knowledge, the existing \textit{MCSP} algorithms are mostly based on linear programming or graph search, and their high complexity prohibits them from scaling to the large networks \cite{jaffe1984algorithms,korkmaz2001multi,shi2017multi}. As a special and simpler case with only one constraint, the \textit{CSP} problem is the basis for solving the \textit{MCSP}, which can be classified into the following categories based on \textit{exact/approximate} and \textit{index-free/index-based} solutions. The \textit{exact CSP} result provides the ground-truth and satisfies the user's actual requirement \cite{hansen1980bicriterion,handler1980dual,LOZANO2013378,SEDENONODA2015602}, but it is extremely slow to compute as a NP-H problem \cite{garey1979computers,hassin1992approximation}. Therefore, the \textit{approximate CSP} algorithms \cite{hansen1980bicriterion,hassin1992approximation,juttner2001lagrange,lorenz2001simple,tsaggouris2009multiobjective} are proposed to reduce the computation time by sacrificing the optimality of the path length. However, their speedup compared with the exact solutions is very limited \cite{kuipers2006comparison,LOZANO2013378}. To further improve the query efficiency, a few \textit{indexes} \cite{storandt2012route,wang2016effective} have been proposed by extending the existing index structures of the shortest path problem. For example, \textit{CSP-CH} \cite{storandt2012route} extends the traditional \textit{CH} index \cite{geisberger2008contraction} to support exact \textit{CSP} query processing. 
Although the index construction time is usually long, which makes them inflexible when the weight or cost changes \cite{zhang2020stream,zhang2021DWPSL,zhang2021DH2H}, the index-based methods enjoy the high efficiency of query answering. Since \textit{2-hop labeling} \cite{cohen2003reachability,akiba2013fast,ouyang2018hierarchy} is the current state-of-the-art shortest path index, in this work, we aim to provide a 2-hop labeling-based index structure for exact \textit{MCSP} queries with both faster index construction and more efficient query processing.

However, it is non-trivial to extend the 2-hop labeling to answer the \textit{MCSP} queries. Since it is already NP-H to find the optimal labels with size of $\varOmega(|V||E|^{1/2})$ for the single-criterion shortest path as proven in \cite{cohen2003reachability}, it would be much harder for the \textit{MCSP} problem due to the following reasons:

\begin{figure}[ht]
    \centering
    \includegraphics[width=3.2in]{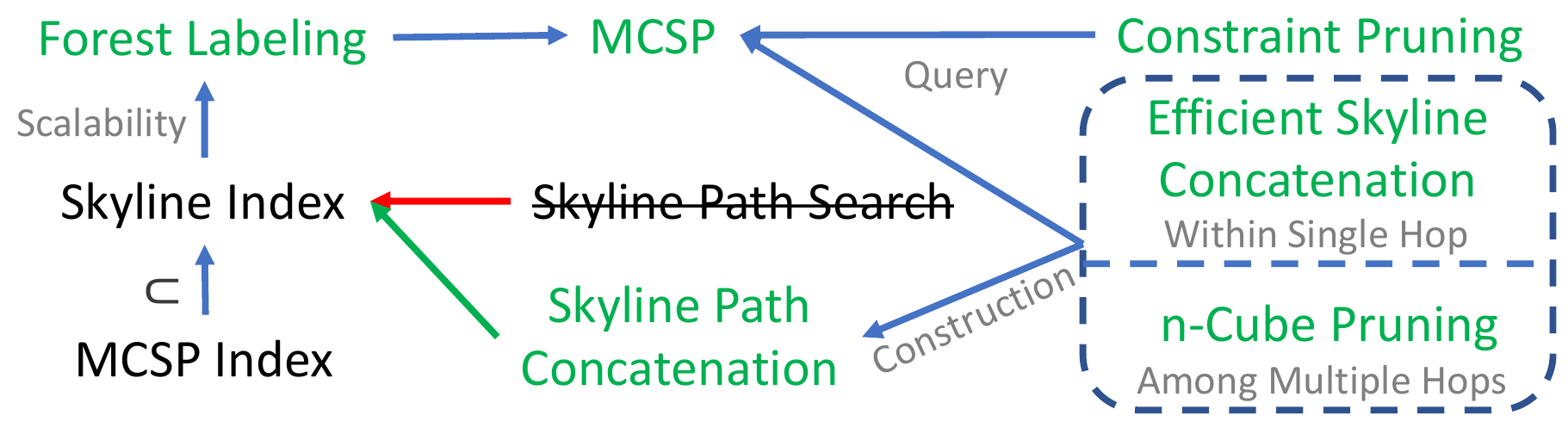}
    \caption{Framework Overview}
    \label{fig:Overview}
\end{figure}

Firstly, since the constraints $C$ can be of any possible combination while are unknown beforehand, the skyline path information is inevitably needed in the index to be able to deal with any scenario. However, skyline path search is very time-consuming as discussed before, and hence the existing \textit{MCSP} solutions all suffer from low efficiency. Consequently, none of the existing indexes are constructed by the exact skyline path but resort to the approximate path search. The \textit{CSP-CH} \cite{storandt2012route} chooses to contract a limited number of vertices and adopts heuristic testing before each skyline path expansion, which leads to a subset of the skyline shortcuts with smaller index size but slower query time. 
\textit{COLA} \cite{wang2016effective} resorts to the graph partitioning and approximates skyline path search to reduce intermediate result size, but it sacrifices the query accuracy. 
Therefore, skyline path search is the bottleneck of existing methods and hinders the development of a full exact index, without which a fast exact \textit{MCSP} query answering is impossible. In this work, we propose a \textit{skyline path concatenation} algorithm to completely avoid the expensive skyline path search and construct a 2-hop-labeling index that stores the complete skyline paths without compromising query accuracy or efficiency. 

Secondly, direct concatenation of two skyline paths does not necessarily generate a skyline path, and another $O(mn\log mn)$ time is needed to validate the skyline results, where $m$ and $n$ are the size of the two skyline path sets to be concatenated. To speed up this process, we provide several insights of the multi-dimensional skyline path concatenation, and propose an efficient concatenation method to incorporate the validation into the concatenation and meanwhile prune the concatenation space. Although the worst case complexity remains the same when all the concatenation results are all skylines, our algorithm is much more efficient in practice.

Thirdly, the skyline path computation can hardly scale to larger networks. As the path length increases in large networks, the number of skyline paths becomes larger. Consequently, both the index size and index construction time raise dramatically. To address the scalability issue, we propose the \textit{forest hop labeling} framework which partitions the networks into multiple regions and constructs the index parallelly both within and across the regions.  


Last but not least, the final index is essentially for skyline path query since it contains all the skyline path information. Hence, it would take approximately $m\times n\times h$ times of skyline concatenations before answering the query, where $h$ is the number of hops. This is much larger than the traditional 2-hop-labeling index which only needs $h$ times of calculation. In other words, the query performance of the index deteriorates dramatically and loses the promise of its high efficiency. Therefore, we propose a $n$-\textit{Cube pruning} technique to prune the useless concatenation as early as possible. Furthermore, we propose a \textit{constraint-based pruning} technique for faster \textit{MCSP} query answering, and also provide pruning technique for the forest label construction. As shown in the experiments, our approach can answer the \textit{MCSP} query within 1 ms and the proposed pruning techniques can further reduce it by two orders of magnitude, while the existing methods are thousands of times slower. Our contributions are illustrated in Figure \ref{fig:Overview} and summarized below:

\begin{itemize}
    \item We first propose a skyline path concatenation-based \textit{MCSP-2Hop} method that avoids the expensive skyline path search to achieve efficient \textit{MCSP} query answering and index construction. Several concatenation insights are further presented to reduce the concatenation space and incorporate multi-dimensional skyline validation into concatenation.
    \item We further propose a \textit{forest hop labeling} for scalability. We also design a $n$-Cube pruning technique to further speed up the multi-hop concatenation during index construction and query answering, and a constraint pruning method to speed up \textit{MCSP} specifically.
    \item We thoroughly evaluate our method with extensive experiments on real-life road networks and the results show that it outperforms the state-of-the-art methods by several orders of magnitude. 
\end{itemize}

This paper extends the work \cite{Liu2021Forest} where we introduced the \textit{Forest Hop Labeling} for \textit{CSP} with only one constraint. Because path concatenation algorithm and pruning technique were only designed for one constraint, its multi-constraint version in \cite{Liu2021Forest} was only a straightforward and slow solution. Therefore, in this paper, we generalize the path concatenation methods into multi-cost scenario and extend the pruning techniques to multi-constraint scenario to achieve higher performance. Although they are the generalizations of the \textit{CSP}, the problem in the multi-dimensional space is much harder than the 2-dimensional one such that \textit{CSP} is only a very special case of \textit{MCSP} with unique properties. Therefore, we provide new insightful properties of the high dimensional skyline path concatenation and pruning to facilitate the more efficient index construction and query answering in \textit{MCSP}. 

In the reminder of this paper, we first define our problem formally and introduce the essential concepts in Section \ref{sec:Preliminary}. Section \ref{sec:MCSP2Hop} and \ref{sec:Forest} present our \textit{MCSP-2Hop} and \textit{forest hop labeling} with their index structures and query answering methods, and Section \ref{sec:Skyline} presents the speedup techniques. We report our experimental results in Section \ref{sec:Experiment}. Finally, Section \ref{sec:RelatedWork} discusses related works and Section \ref{sec:Conclusion} concludes the paper. 

\section{Preliminary}
\label{sec:Preliminary}

\subsection{Problem Definition}
\label{subsec:Preliminary_PD}
A road network is a $n$-dimensional graph $G(V, E)$, where $V$ is a set of vertices and $E\subseteq V\times V$ is a set of edges. Each edge $e \in E$ has $n$ criteria falling into two categories: weight $w(e)$ and a set of costs $\{c_i(e)\}$, $i\in[1,n-1]$. A path $p$ from the source $s \in V$ to the target $t \in V$ is a sequence of consecutive vertices $p=\left<s=v_0,v_1,\dots,v_k=t\right>$ = $\mycolor{\left<e_0,e_1,\dots,e_{k-1}\right>}$, where $e_i$= $(v_i,v_{i+1})\in E, \forall i\in[0,k-1]$. Each path $p$ has a weight $w(p)$=$\sum_{e \in p}w(e)$ and a set of costs $\{c_i(p)$= $\sum_{e \in p}c_i(e)\}$. We assume the graph is undirected while it is easy to extend our techniques to the directed one. Then we define the \textit{MCSP} query below:

\begin{definition}[MCSP Query]
\label{def:MCSP}
Given a $n$-dimensional graph $G(V, E)$, a $MCSP$ query $q(s,t,\overline{C})$ returns a path with the minimum $w(p)$ while each $c_i(p) \leq C_i\in \overline{C}$. 
\end{definition}

As analyzed in Section \ref{sec:Introduction}, it is quite slow to answer a \textit{MCSP} query by online graph search. Therefore, in this paper, we resort to index-based method and study the following problem: 

\begin{definition}[Index-based MCSP Query Processing]
\label{def:PD}
Given a graph $G$, we aim to construct a hop-based index $L$ that can efficiently answer any \textit{MCSP} query $q(s,t,\overline{C})$ only with $L$, \mycolor{which is stored in a lookup table and can avoid the online searching in road networks.}
\end{definition} 

Unlike the index for the shortest path queries which only needs to store the shortest distance from one vertex to another, the \textit{MCSP} index needs to cover all the possible constraints of each criteria $C_i$. In addition, we observe that the \textit{MCSP} paths are essentially a subset of the \textit{multi-dimensional skyline} paths, which is introduced as follows.




\subsection{Multi-Dimensional Skyline Path}
\label{subsec:Preliminary_Skyline}
Like the \textit{dominance} relation of the skyline query in the 2-dimensional space \cite{papadias2003optimal}, we first define the \textit{dominance} relation between any two paths with the same source and destination:

\begin{definition}[Multi-Dimensional Path Dominan -ce]
	\label{def:Multi-Dominance}
	Given two paths $p_1$ and $p_2$ with the same $s$ and $t$, $p_1$ dominates $p_2$ iff $w(p_1) \leq w(p_2)$ and $c_i(p_1) \leq c_i(p_2),\forall i\in[1,n-1]$, and at least one of them is the strictly smaller relation.
\end{definition}

\begin{figure}[ht]
    \centering
    \includegraphics[scale=0.3]{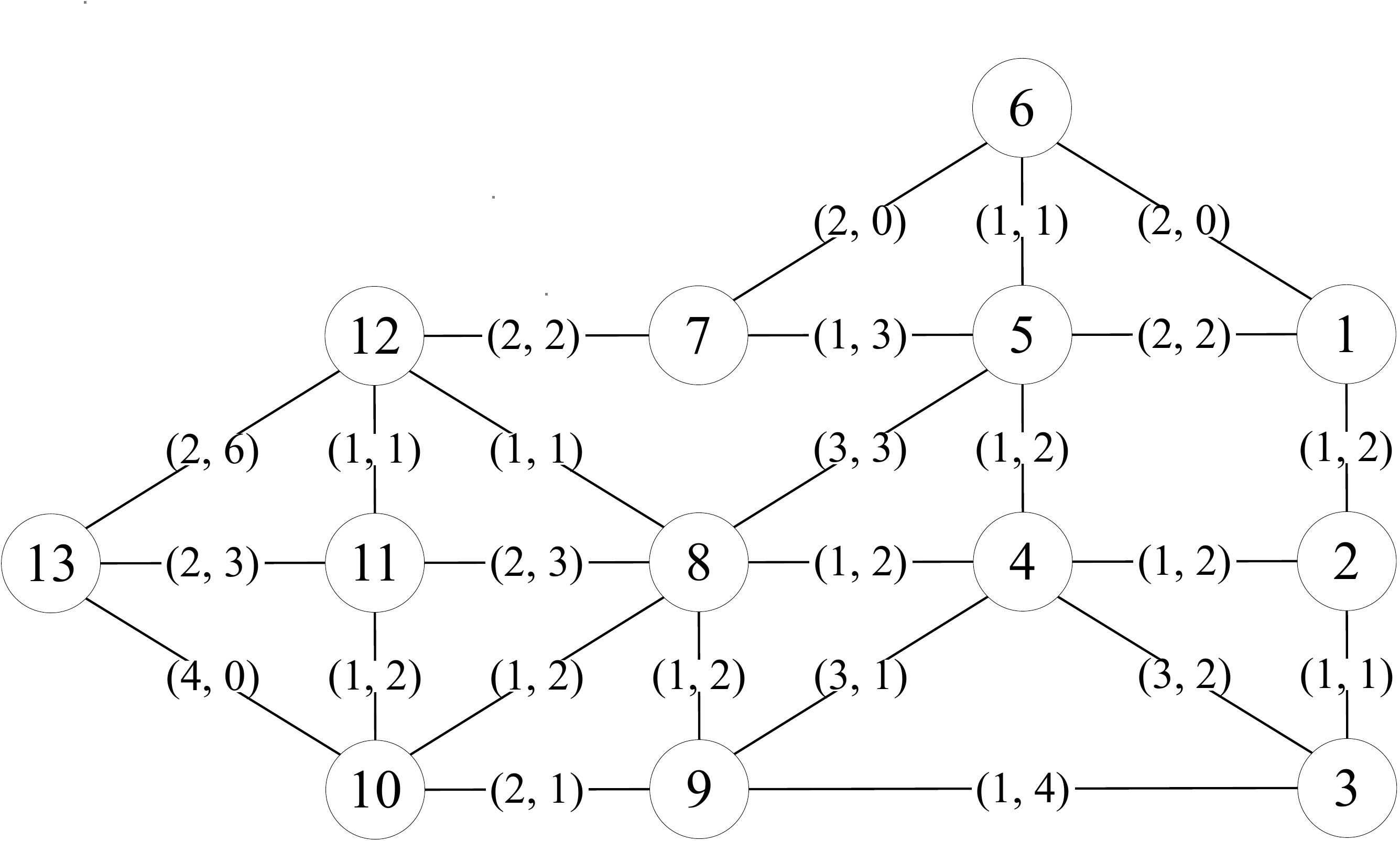}
    \caption{An Example of Skyline Path in 2D Graph. Edge Labels are $(w(e), c_1(e))$.}
    \label{fig:ExampleGraph}
\end{figure}


\begin{definition}[Multi-Dimensional Skyline Path]
	\label{def:Multi-SkylinePath}
	Given two paths $p_1$ and $p_2$ with the same $s$ and $t$, they are skyline paths if they cannot dominate each other. 
\end{definition}


In this paper, \mycolor{we use $P(s,t)$ to denote the set of skyline paths from $s$ to $t$. We denote a path $p$ as a skyline path $\overline{p}$ if it is in a skyline path set}. Consider the 2-dimensional example in Figure \ref{fig:ExampleGraph}, we enumerate four paths from $v_5$ to $v_9$:
\begin{enumerate}
    \item $p_1 = \left<v_5, v_4, v_9\right>$, $w(p_1)=4$, $c_1(p_1)=3$
    \item $p_2 = \left<v_5, v_8, v_9\right>$, $w(p_2)=4$, $c_1(p_2)=5$
    \item $p_3 = \left<v_5, v_4, v_8, v_9\right>$, $w(p_3)=3$, $c_1(p_3)=6$
    \item $p_4 = \left<v_5, v_8, v_4, v_9\right>$, $w(p_3)=7$, $c_1(p_3)=6$
\end{enumerate}
According to the definition, we can identify that $p_2$ is dominated by $p_1$, and $p_4$ is dominated by the other three paths. Hence, $p_2$ and $p_4$ are not skyline paths. Meanwhile, as $p_1$ and $p_3$ cannot dominate each other, we obtain two skyline paths from $v_5$ to $v_9$: $P(v_5,v_9)=\mycolor{\{\overline{p}_1,\overline{p}_3\}}$. Next we prove that skyline paths are essential for building any \textit{MCSP} index:

\begin{theorem}
\label{the:Skyline}
The skyline paths between any two vertices are complete and minimal for all the possible \textit{MCSP} queries.
\end{theorem}
\begin{proof}
We first prove the completeness. Suppose the skyline paths from $s$ to $t$ are $\{\overline{p}_1$, $\dots$, $\overline{p}_k\}$, then each cost dimension  $c_i$ can be divided into $k+1$ intervals: $[0,c_i^1),[c_i^1,$ $c_i^2),\dots, [c_i^k,\infty)$, where each $c_i^j$ is the $j^{th}$ cost value of criterion $c_i$ sorted increasingly, and the entire space can be decomposed into $(k+1)^{(n-1)}$ subspaces. Then for the subspaces with any interval falling into $[0,c_i^1)$, there is no valid path satisfying all the constraints at the same time. For the subspaces with no interval falling into $[0,c_i^1)$, there is always at least one path dominating all the constraints in it, and the one with the smallest distance is the result.

Next we prove this skyline path set is minimal. Suppose we remove any skyline path $\overline{p}_j=\{c_i^j\}$ from the skyline result set. Then the \textit{MCSP} query whose constraint is within the subspace of $[c_1^{j}, c_1^{j+1})\times [c_2^{j}, c_2^{j+1})\times \dots \times [c_k^{j}, c_k^{j+1})$ will have no valid path. Hence, $\overline{p}_j$ cannot be removed and the skyline path set is minimal for all the \textit{MCSP} queries between $s$ and $t$.
\end{proof}

In summary, building an index for \textit{MCSP} query is equivalent to building the index for multi-dimensional skyline path query.

\subsection{2-Hop Labeling}
\label{subsec:Preliminary_2Hop}
\textit{2-hop labeling} is an index method that does not involve any graph search during query answering. It answers the shortest distance query by table lookup and summation. Specifically, $\forall u \in V$, we assign a label set $L(u)={\{v,w(v,u)\}}$ to store the minimum weights between a set of \textit{hop vertices} $v$ to $u$. To answer a query $q(s,t)$, we first determine the intermediate hop set $H=L(s) \cap L(t)$. Then we can find the minimum weight via the hop set: $min_{h \in H}{\{w(s,h)+w(h,t)\}}$. 
In this work, we select \textit{H2H} \cite{ouyang2018hierarchy}, the state-of-the-art 2-hop labeling, to support \textit{MCSP} query since it can avoid graph traversal in both index construction and query processing. Specifically, \textit{Tree Decomposition} \cite{robertson1986graph} is the cornerstone of \textit{H2H}, which maps a graph into a tree structure. We introduce it briefly as follows.

\begin{definition}[Tree Decomposition]
	\label{def:TreeDecomposition}
	Given a graph $G$, its tree decomposition $T$ is a rooted tree in which each node $X\in V$ is a subset of $V$. $T$ has the following tree properties:
	\begin{enumerate}
		\item $\bigcup_{X_i\in V}X_i=V$;
		\item $\forall (u, v)\in E, \exists X_i$ such that $\{u, v\}\subseteq X_i$;
		\item $\forall v\in V$, the set $\{X|v\in X\}$ forms a subtree of $T$.
	\end{enumerate}
\end{definition}

\begin{figure}[htbp]
	\centering
	\includegraphics[width=2.6in]{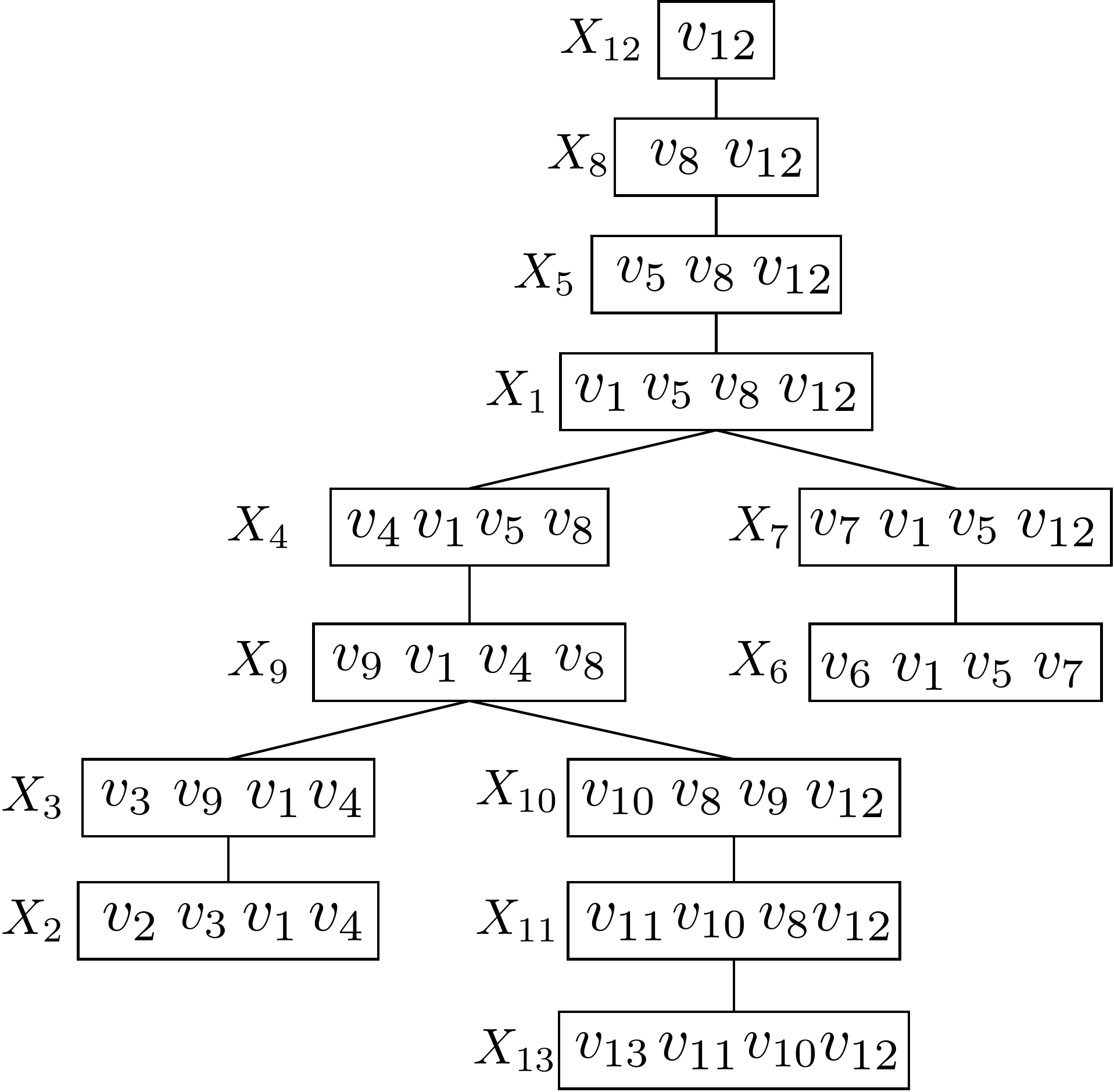}
	\caption{Tree Decomposition Example}
	\label{fig:Tree Decomposition}
\end{figure}

For instance, Figure \ref{fig:Tree Decomposition} shows the tree decomposition of the example graph in Figure \ref{fig:ExampleGraph}. The first vertex in each tree node is its representative vertex. 
This structure has a \textit{Cut Property} that can be used to assign labels: $\forall s,t\in V$ and $X_s,X_t$ are their corresponding tree nodes without an ancestor/descendent relation, their cuts are in their \textit{lowest common ancestor (LCA)} node \cite{chang2012exact}, \mycolor{which is the lowest tree node that has both $X_s$ and $X_t$}. For example, the tree nodes of $v_2$ and $v_7$ are $X_2$ and $X_7$, and their \textit{LCA} is $X_{1}=\{v_{1},v_5,v_8,v_{12}\}$. Then these four vertices form a cut between $v_2$ and $v_7$. If $s$ is $t$'s descendant, then $t$ can be viewed as a cut between itself and $s$. Therefore, the \textit{LCA} of a vertex pair contains all the cuts between them. 
Accordingly, we assign labels to a vertex with all its ancestors in $T_G$. It should be noted that the vertices in each tree node are also in its label because of the \textit{Ancestor Property}: $\forall X_v\in T_G, X_u(u\in X_v\backslash\{v\})$ is an ancestor of $X_v$ in $T_G$. For example, $v_6$ has labels of $\{v_7,v_1,v_5,v_8,v_{12}\}\supseteq \{v_1,v_5,v_7\}$.

\section{MCSP-2Hop Labeling}
\label{sec:MCSP2Hop}

In this section, we first present the \textit{MCSP-2Hop} index structure and query processing, followed by the index construction, correctness proof and path retrieval.

\begin{table*}
\centering
\caption{MCSP-2Hop Labels of 2D Graph}
\label{table:MCSP-2Hop}
\scriptsize
\begin{tabular}{|c|c|c|c|c|c|c|c|c|c|c|}
\hline
\textbf{} & \textbf{$v_3$} & \textbf{$v_7$} & \textbf{$v_{11}$} & \textbf{$v_{10}$} & \textbf{$v_9$}                                                 & \textbf{$v_4$}                                                        & \textbf{$v_1$}                                                                    & \textbf{$v_5$}                                                                   & \textbf{$v_8$}                                                        & \textbf{$v_{12}$}                                              \\ \hline
\textbf{$v_2$}     & (1, 1)         &                &                   &                   & (2, 5),(4, 3)                                                  & (1, 2)                                                                & (1, 2)                                                                            & (2, 4),(4, 3)                                                                    & (2, 4)                                                                & (3, 5),(7, 4)                                                  \\ \hline
\textbf{$v_3$}     &                &                &                   &                   & (1, 4),(6, 3)                                                  & (2, 3),(3, 2)                                                         & (2, 3)                                                                            & (3, 5),(4, 4)                                                                    & \begin{tabular}[c]{@{}l@{}}(2, 6),(3, 5)\\ (4, 4)\end{tabular}        & \begin{tabular}[c]{@{}l@{}}(3, 7),(4, 6)\\ (5, 5)\end{tabular} \\ \hline
\textbf{$v_6$}     &                & (2, 0)         &                   &                   &                                                                &                                                                       & (2, 0)                                                                            & (1, 1)                                                                           & \begin{tabular}[c]{@{}l@{}}(3, 5),(4, 4)\\ (5, 3)\end{tabular}        & (4, 2)                                                         \\ \hline
\textbf{$v_7$}     &                &                &                   &                   &                                                                &                                                                       & (3, 5),(4, 0)                                                                     & (1, 3),(3, 1)                                                                           & (3, 3)                                                                & (2, 2)                                                         \\ \hline
\textbf{$v_{13}$}  &                &                & \begin{tabular}[c]{@{}c@{}}(2, 3)\\(5, 2)\end{tabular}     & \begin{tabular}[c]{@{}c@{}}(3, 5)\\(4, 0)\end{tabular}     & \begin{tabular}[c]{@{}c@{}}(4, 9),(5, 6)\\ (6, 1)\end{tabular} & \begin{tabular}[c]{@{}c@{}}(4, 9),(5, 7)\\ (6, 4),(9, 2)\end{tabular} & \begin{tabular}[c]{@{}l@{}}(6, 13),(7, 11)\\ (8, 8),(9, 6)\\ (12, 5)\end{tabular} & \begin{tabular}[c]{@{}l@{}}(5, 11),(6, 9)\\ (7, 6),(8, 5)\\ (10, 4)\end{tabular} & \begin{tabular}[c]{@{}l@{}}(3, 7),(4, 5)\\ (5, 2)\end{tabular}        & \begin{tabular}[c]{@{}l@{}}(2, 6),(3, 4)\\ (6, 3)\end{tabular} \\ \hline
\textbf{$v_{11}$}  &                &                &                   & (1, 2)            & (3, 3)                                                         & (3, 4)                                                                & \begin{tabular}[c]{@{}l@{}}(5, 8)\\ (7, 3)\end{tabular}                           & \begin{tabular}[c]{@{}l@{}}(4, 6),(5, 5)\\ (6, 4)\end{tabular}                   & (2, 2)                                                                & (1, 1)                                                         \\ \hline
\textbf{$v_{10}$}  &                &                &                   &                   & (2, 1)                                                         & (2, 4),(5, 2)                                                         & \begin{tabular}[c]{@{}l@{}}(4, 8),(6, 7)\\ (7, 6),(8, 5)\end{tabular}             & \begin{tabular}[c]{@{}l@{}}(3, 6),(4, 5)\\ (6, 4)\end{tabular}                   & (1, 2)                                                                & (2, 3)                                                         \\ \hline
\textbf{$v_9$}     &                &                &                   &                   &                                                                & (2, 4),(3, 1)                                                         & (3, 7),(5, 5)                                                                     & (3, 6),(4, 5)                                                                    & (1, 2)                                                                & (2, 3)                                                         \\ \hline
\textbf{$v_4$}     &                &                &                   &                   &                                                                &                                                                       & (2, 4),(4, 3)                                                                     & (1, 2)                                                                           & (1, 2)                                                                & (2, 3)                                                         \\ \hline
\textbf{$v_1$}     &                &                &                   &                   &                                                                &                                                                       &                                                                                   & (2, 2),(3, 1)                                                                    & \begin{tabular}[c]{@{}l@{}}(3, 6),(5, 5)\\ (6, 4),(7, 3)\end{tabular} & (4, 7),(6, 2)                                                  \\ \hline
\textbf{$v_5$}     &                &                &                   &                   &                                                                &                                                                       &                                                                                   &                                                                                  & (2, 4),(3, 3)                                                         & \begin{tabular}[c]{@{}l@{}}(3, 5),(4, 4)\\ (5, 3)\end{tabular} \\ \hline
\textbf{$v_8$}     &                &                &                   &                   &                                                                &                                                                       &                                                                                   &                                                                                  &                                                                       & (1, 1)                                                         \\ \hline
\textbf{$v_{12}$}  &                &                &                   &                   &                                                                &                                                                       &                                                                                   &                                                                                  & (1, 1)                                                                &                                                                \\ \hline
\end{tabular}
\end{table*}

\subsection{Index Structure}
\label{subsec:MCSP2Hop_Index}
Unlike the traditional 2-hop labeling whose label is only a vertex and its corresponding distance value, \textit{MCSP-2Hop}'s label is a set of skyline paths. Specifically, for each vertex $u \in V$, it has a label set $L(u)=\{(v,$ $P(u,v))\}$, where $v$ is the hop vertex ($\{v\}$ is $u$'s hop set) and $P(u,v)$ is the skyline path set. We use $L=\{(L(u)|\forall u \in V\}$ to denote the set of all the labels. If $L$ can answer all the \textit{MCSP} queries in $G$, then we say $L$ is a \textit{MCSP-2Hop} cover. The detailed labels of the example graph in Figure \ref{fig:ExampleGraph} are shown in Table \ref{table:MCSP-2Hop}, with the rows being the vertices and columns being the hops. With the help of the tree decomposition in Figure \ref{fig:Tree Decomposition}, we can see each vertex stores all its ancestors on tree as labels. For example, $X_{12}, X_8$ and $X_5$ are the ancestors of $X_1$, then they appear in $v_1$'s labels. Because $X_2, X_{13}$, and $X_6$ are the leaves, they will not appear in any vertex's labels. Some cells have several values because they are skyline paths. The details of the index construction will be presented in Section \ref{subsec:MCSP2Hop_Construction}.

\subsection{Query Processing}
\label{subsec:MCSP2Hop_Query}
Given a \textit{MCSP} query $q(s,t,\overline{C})$, we can answer it with the \textit{MCSP-2Hop} in the same way like the single-criteria scenario:
we first determine the common intermediate hop set $H=L(s) \cap L(t)=\{h|h\in X(LCA(s,t)\}$ by retrieving the vertices in the \textit{LCA} of $s$ and $t$, where \textit{LCA}$(s,t)$ returns the \textit{LCA} of $X_s$ and $X_t$. Then for each $h_i\in H$, we compute its candidate set $P_c(h_i)$ by concatenating the paths $P(s,h_i)$ and $P(h_i,t)$. By taking the union of $P_c(h_i)$ from all hops in $H$, we obtain the candidate set $P_c(H)$:
\begin{equation*}
\begin{aligned}
P_c(H) &= \{P_c(h_i)|
\{P(s,h_i)\oplus P(h_i,t)\}\}, \forall h_i\in H\\
&=\bigcup\{ p_j\oplus p_k\},\forall p_j\in P(s,h_i) \wedge p_k\in P(h_i,t) 
\end{aligned}
\end{equation*}
where $\oplus$ is \textit{path concatenation} operator that takes the sum of paths' weights and the corresponding costs. Each $P_c(h_i)$ contains all the pairwise skyline concatenation results, and $P_c(H)$ is their union.

During the concatenation, we maintain the current optimal result (i.e., the shortest one satisfying the constraint $\overline{C}$) and keep updating it. Finally, after all the candidates are generated and validated, we return the optimal one as the \textit{MCSP} result. For example, given a query $q(v_9,v_7,\{6\})$, the \textit{LCA} of $v_9$ and $v_7$ in the tree is $X_1$, and it contains vertices $H=\{v_1,v_5,v_8,v_{12}\}$. For hop $v_1$, $L(v_9)[v_1]=\{(3,7),(5,5)\}$ and $L(v_7)[v_1]=\{(3,5),(4,0)\}$. Therefore, we can get four candidates $P_c(v_1)=\{(6,12),(7,7),(8,10),(9,5)\}$. Only $(9,5)$ has a cost smaller than $6$, so $(9,5)$ is the current best result. For hop $v_5$, we can get candidates $P_c(v_5)=\{(4,9),(5,6),$ $(6,7),(7,4)\}$. Then $(5,6)$ replaces $(9,5)$ as it has smaller weight and meet the constraint condition as well. For hop $v_8$, we can get candidates $P_c(v_8)=\{(4,5)\}$, and the current best result is updated to $(4,5)$. Finally, for hop $v_{12}$, we obtain the candidate $P_c(v_{12})=\{(4,5)\}$, which is equal to the current best result. The final result is $(4,5)$. 

As shown in the above example, although only 4 hops are needed, we actually run 10 path concatenation operations. Specifically, unlike the 2-hop labeling of the single-criterion case where only $|H|$ times of adding is needed, the concatenation of the skyline paths makes hop computation much more time-consuming. Suppose the label of $s$ to $h$ has $m$ skyline paths and the label of $h$ to $t$ has $n$ skyline paths. Then for each intermediate hop, we have to compute $mn$ times, and the total time complexity grows to $O(mn|H|)$ instead of $O(|H|)$. Therefore, pruning over this large amount of computation is crucial to the query performance and we will elaborate it in Section \ref{sec:Skyline}.

\subsection{Index Construction}
\label{subsec:MCSP2Hop_Construction}
The index construction is made up of two phases: a bottom-up \textit{Skyline Tree Decomposition} that gathers the skyline shortcuts, creates the tree nodes and forms the tree, and a top-down \textit{Skyline Label Assignment} that populates the labels with query answering, \mycolor{as shown in Algorithm \ref{Algorithm:Construction}}. \mycolor{It should be noted that both of these two phases do not involve the expensive skyline path search and only use the skyline path concatenation.}
\begin{algorithm}[ht]
	\caption{MCSP-2Hop Construction}
	\label{Algorithm:Construction}
	\LinesNumbered
	\KwIn{Graph $G(V,E)$}
	\KwOut{MCSP-2Hop Labeling $L$}
	//	Tree Node Contraction\\
	\While{$v\in V$ has the minimum degree and $V\neq \phi$}
	{
		$X_v=N(v)$, $r(v)\leftarrow$ Iteration Number\\
		\For{$(u,w)\leftarrow N(v)\times N(v)$}
		{
			$P(u,w)\leftarrow Skyline(p(u,v)\oplus p(v,u), (u,w))$\;
			$E\leftarrow E\cup (P(u,w))$, $V\leftarrow V-v$, 	$E\leftarrow E-(u,v)-(v,w)$\;
		}
	}
	// Tree Formation\\
	\For{$v$ in $r(v)$ increasing order}
	{
		$u\leftarrow min\{r(u)|u\in X(v)\}, X_v.Parent\leftarrow X_u$
	}
	// Label Assignment\\
	$X_v\leftarrow$ Tree Root, $Q.insert(X_v)$\\
	\While{$X_v\leftarrow Q.pop()$}
	{
		\For{$u\in \{u|X_u$ is ancestor of $X_v\}$}
		{
			$P(v,u))\leftarrow Skyline(P(v,w)\oplus P(w,u)), \forall w\in X_v$\;
			$L(v)\leftarrow{(u,P(v,u))}$\;
		}
		$Q.insert(X_u.children)$, $L\leftarrow L\cup L(v)$\;
	}
\end{algorithm}
\subsubsection{\mycolor{Skyline Tree Decomposition}}
\label{subsubsec:CSP2Hop_TD_Contraction}
We contract the vertices in the degree-increasing order, which is created by \textit{Minimum Degree Elimination}\cite{gong2019skyline,li2019scaling}. For each $v\in V$, we add a skyline shortcut $p'(u,w)$ to all of its neighbor pairs $(u,w)$ by $p(u,v)\oplus p(v,w)$. If $(u,w)$ does not exist, then we use $p'_s(u,w)$ as its skyline path $p(u,w)$ directly. Otherwise, we compute the skyline of $p(u,w)$ and $p'(u,w)$ in a linear time sort-merge way, \mycolor{which sorts the two path sets based on $w$ first and obtain the new skyline paths by comparing them in the distance-increasing order}. After that, a tree node $X(v)$ is formed with $v$ and its neighbors $N(v)$, and all their corresponding skyline paths from $v$ to $N(v)$. The tree node $X(v)$ is assigned an order $r(v)$ according to the order when $v$ is contracted. After contracting $v$, $v$ and its edges are removed from $G$ and the remaining vertices' degree are updated. Then we contract the remaining vertices in this way until none remains. \mycolor{The time complexity of the contraction phase is $O(|V|(max(|X_i|)^{2}\cdot n\cdot  c^{2(n-1)}_{max}\log c_{max}+\log |V|))$, where $max(|X_i|)$ is the tree width, $c_{max}$ is the longest edge cost (including weights) and $c_{max}^{n-1}$ determines the worst case of skyline path number, and $c^{2(n-1)}_{max}\log c^{2(n-1)}_{max}$ is concatenation complexity.}

We use the graph in Figure \ref{fig:ExampleGraph} as an example. Firstly, we contract $v_2$ with $N(v_2)=\{v_1,v_3,v_4\}$. Because there is no edge between $v_1$ and $v_3$, we add a shortcut $(v_1,v_3)$ with  by $w=w(v_1,v_2)+w(v_2,v_3)=2$ and $c=c(v_1,v_2)+c(v_2,v_3)=3$. Similarly, there is no edge between $v_1$ and $v_4$, so we add a new edge $(v_1,v_4)$ with $w=w(v_1,v_2)+w(v_2,v_4)=2$ and $c=c(v_1,v_2)+c(v_2,v_4)=4$. As for $v_4$ and $v_3$, we first generate a shortcut by $p(v_4,v_2)\oplus p(v_2,v_3)=(2,3)$. Then because there already exists an edge $(3,2)$ from $v_4$ to $v_3$, we compute the new skyline paths between the shortcut and the edge, and get the new skyline paths $p(v_4,v_3)=\{(3,2),(2,3)\}$. After that, we remove $v_2$ from the graph, leaving $v_1$'s degree changed to 4, $v_4$'s degree remaining 5, and $v_3$'s degree remaining 3. We also obtained a tree node $X(v_2,v_3,v_1,v_4)$ with the skyline paths from $v_2$ to $N(v_2)$ and order $r(v_2)=1$. This procedure runs on with the next vertex with the smallest degree. When all the vertices finished contraction, we have an order of $:v_2,v_3,v_6,v_7,v_{13},v_{11},$\\$v_{10},v_9,v_4,v_1,v_5,v_8,v_{12}$.

After that, we connect all the nodes to form a tree by connecting each tree node to its smallest order neighbor. For example, $X_{10}$ has a neighbor set of $\{v_9,v_9,v_{12}\}$ and $X_9$ has the smallest order of them, so $X_{10}$ is connected to $X_9$ as its child node.

\subsubsection{Skyline Label Assignment}
\label{subsubsec::CSP2Hop_TD_Label}
It should be noted the obtained tree preserves all the skyline path information of the original graph because every edge's information are preserved in the shortcuts. In fact, we can view it as a \textit{superset} of the \textit{MCSP-CH} result. Therefore, we can assign the labels with skyline search on this tree directly. However, it would be very time-consuming as the \textit{Sky-Dijk} is much slower than \textit{Dijkstra}'s. To this end, we use the top-down label cascading assignment method that takes advantages the existing labels with query answering to assign labels to the new nodes.

Specifically, a label $L(u)$ stores all the skyline paths from $u$ to its ancestors. Therefore, we start from the root of the tree to fill all the labels. For root itself, it has no ancestor so we just skip it. Next for its child $u$, because the root is the only vertex in $X(u)$ other than $u$, we add the root to $u$'s label. Then for any tree node $X(v)$, it is guaranteed that all its ancestors have obtained their labels. Then for any of its ancestor $a_i$, we compute its skyline paths by taking the skyline paths from  $\{p(v,u)\oplus p(u,a_i)| \forall u\in X(v)$\}. The time complexity of the label assignment phase is $O(|V|max(|X_i|) \cdot n\cdot c^{2n-2}_{max}\log c_{max})$, where $h$ is the tree height. The space complexity of the labels is $O(h\cdot |V|\cdot c^{n-1}_{max})$.


For example in Figure \ref{fig:Tree Decomposition}, suppose we have already obtained the labels of $X_1, X_5, X_8$, and $X_{12}$, and we are computing the labels of $X(v_4)$, which has an ancestor set of $\{v_{12}, v_8, v_5, v_1\}$. We take the computation of $L(v_4)[v_{12}]$ as example. $X_4$ has three vertices $\{v_1,v_5,v_8\}$ apart from $v_4$, so we compute the skyline paths via these tree hops. To show the process clearly, we do not apply the pruning first.
\begin{equation*}
    \begin{aligned}
        P_c(h_1)&=p(v_4,v_1)\oplus p(v_1,v_{12})\\
        &=\{(2,4),(4,3)\}\oplus\{(4,7),(6,2)\}\\
        &=\{(6,11),(8,6),(8,10),(10,5)\}\\  
        P_s(h_1)&=\{(6,11),(8,6),(10,5)\}\\
    \end{aligned}
\end{equation*}
\begin{equation*}
    \begin{aligned}
        P_c(h_5)&=p(v_4,v_5)\oplus p(v_5,v_{12})\\
        &=\{(1,2)\}\oplus \{(3,5),(4,4),(5,3)\}\\
        &=\{(4,7),(5,6),(6,5)\}\\
        P_s(h_5)&=\{(4,7),(5,6),(6,5)\}\\
    \end{aligned}
\end{equation*}
\begin{equation*}
    \begin{aligned}
        P_c(h_8)&=p(v_4,v_8)\oplus p(v_8,v_{12})\\
        &=\{(1,2)\}\oplus\{(1,1)\}=\{(2,3)\}\\
        P_s(h_8)&=\{(2,3)\}
    \end{aligned}
\end{equation*}

Because $(2,3)$ dominates all the others, we put $(2,3)$ into the label of $L(v_4)[v_{12}]$. 

\subsection{Correctness}
\label{subsec:MCSP2Hop_Correctness}
In this section, we show that the \textit{MCSP-2Hop} constructed in Algorithm \ref{Algorithm:Construction} is correct, i,e., the \textit{MCSP} query between any source vertex $s$ and target vertex $t$ can be correctly answered.

As proved in Theorem \ref{the:Skyline}, the skyline paths between $s$ and $t$ are enough to answer any \textit{MCSP} query, so we only need to prove the candidate set $P_c(H)$ covers all the skyline paths between $s$ and $t$. Firstly, we introduce the following lemma:  
\begin{lemma}
	\label{lemma:SkylineSubpath}
	If a path $p$ is a skyline path, the sub-paths on $p$ are all skyline paths.
\end{lemma}
\begin{proof}
	The proof of it can be found in \cite{storandt2012route}.
\end{proof}

\begin{theorem}
	\label{the:HopSkyline}
	$P_c(h)$ covers all the skyline paths from $s$ to $t$ that travel via hop $h$.
\end{theorem}
\begin{proof}
	Suppose $p^*$ is a skyline path from $s$ to $t$ via $h$ and $p^*\notin P_c(h)$. According to Lemma \ref{lemma:SkylineSubpath}, $p^*$'s sub-paths $p^*(s,h)$ and $p^*(h,t)$ are also skyline paths from $s$ to $h$ and from $h$ to $t$. If $p^*(s,h)\notin P(s,h)$ or $p^*(h,t)\notin P(h,t)$, it contradicts the definition that $P(s,h)$ contains all the skyline paths from $s$ to $h$, and similar for $P(h,t)$. Therefore, $P_c(h)$ covers all the skyline paths from $s$ to $t$ via $h$.
\end{proof}

Next we need to prove the union of all the hops $h\in H$ covers all the skyline paths between $s$ and $t$. Because our 2-hop labeling is based on the graph cut, we only need to prove the following theorem:
\begin{theorem}
	\label{the:CSP-2Hop}
	If the hop set $H$ is the set of cuts between $s$ and $t$, then $P_c(H)=\bigcup_{h_i\in H} P_c(h_i)$ covers all the skyline paths from $s$ to $t$.
\end{theorem}
\begin{proof}
	Similar to Theorem \ref{the:HopSkyline}, we also prove it with contradiction. Suppose $p^*$ is a skyline path and $p^*\notin P_c(H)$. Because $H$ is the graph cut set, $p^*$ has to travel through one of the cut vertex $h_i\in H$, which has a candidate set $P_c(h_i)$. However, $p^*\notin P_c(h_i)$ contradicts with Theorem \ref{the:HopSkyline}. So every skyline path from $s$ to $t$ exists in $P_c(H)$.
\end{proof}

\mycolor{Finally, we prove the labels generated by Algorithm \ref{Algorithm:Construction} is a \textit{MCSP-2Hop}.}
\begin{theorem}
	\mycolor{$L$ constructed by Algorithm \ref{Algorithm:Construction} is \textit{MCSP-2Hop}, that is $\forall (v,P(u,v))\in L(u), (u, v\in V)$, the path set $P(u,v)$ covers skyline paths between $u,v$.}
\end{theorem}
\begin{proof}
	\mycolor{Firstly, we define graph $G(v)$ as the assembly of tree nodes from $X_v$ to the root of \textit{skyline tree decomposition} $T_G$. It can be easily obtained that $G(v)$ is \textit{skyline preserved} by referring to the \textit{distance preserved property} of tree decomposition in \cite{ouyang2018hierarchy}. Next, we prove that $P(u,v)$ covers skyline paths between $u,v$ through induction. We suppose that height $h$ of tree root is zero and the height of one tree node is its parent's height plus one. For the root with $h=0$, its label are \textit{MCSP-2Hop}. For $X_v$ with height $h=1$, the labels in $L(v)$ is also \textit{MCSP-2Hop} because of the skyline preserved property. Then suppose that labels in $L(v)$ with height $h=i (i>1)$ is \textit{MCSP-2Hop}, we attempt to prove that labels in $L(u)$ with $h=i+1$ are also \textit{MCSP-2Hop}. Since $G(u)$ is skyline preserved with $X_u\setminus \{u\}$ being $u's$ neighbor set in $G(u)$, then all the skylines paths must pass through vertices in this set. According to the top-down label assignment, the labels is obtained by concatenating the two paths with $X_u\setminus \{u\}$ being the hops, so the computed labels in $L(u)$ are also \textit{MCSP-2Hop}. Therefore, $L$ constructed by Algorithm \ref{Algorithm:Construction} is \textit{MCSP-2Hop}.}
\end{proof}

 
\subsection{Path Retrieval}
\label{subsec:MCSP2Hop_Path}
To retrieve the actual result path, we need to preserve the intermedia vertex in the labels. Specifically, during the \textit{Tree Node Contraction} phase, we need to store the contracted vertex of each skyline path in each shortcut. During the path retrieval, suppose one of the skyline results is the concatenation of $s\rightarrow h$ and $h\rightarrow t$, then they are actually two shortcuts so we only need to insert the intermediate vertices recursively until all no shortcut remains to recover the original path.

\section{Forest Hop Labeling}
\label{sec:Forest}
Due to the nature of skyline, the \textit{MCSP-2Hop} size could be orders of magnitude larger than the distance labels, especially when the paths are long. Therefore, we present our \textit{Forest Hop Labeling} to reduce the index size by partitioning the graph into regions and building a small tree for each region, and a boundary tree between regions.

\subsection{Graph Partition}
\label{subsec:Forest_Parition}
We partition $G$ into a set $C=\{G_1,G_2,\dots,G_{|C|}\}$ of vertex-disjoint subgraphs such that $\bigcup_{i\in[1,|C|]}G_i=G$. If an edge appears in two different subgraphs, it is a \textit{Boundary Edge}, and its end vertices are called \textit{Boundary Vertex}. Graph partitioning is a well-studied problem and we could use any existing solutions. In our implementation, we use the state-of-art \textit{Natural Cut} method \cite{delling2011graph}, because such an approach can create only a small number of boundary vertices.

\subsection{Inner Partition Tree}
\label{subsec:Forest_Inner}
Conceptually, we build a small tree (\textit{MCSP-2Hop}) for each partition and these procedures can run in parallel. However, these trees cannot be constructed directly with only the subgraph information because the actual paths may take a detour out of the current partition and come back later. Therefore, we need to pre-compute the all-pair boundary results for each partition because these boundaries form the cut of a partition naturally and their all-pair results can cover all the detours. In other word, a complete graph of the boundaries is overlapped to the partition with each edge is a skyline path set. This procedure is also fast because we only need to run several search space reduced \textit{Sky-Dijk}s in parallel. When forming the tree nodes, we also follow the degree elimination order. However, we postpone the boundary contraction to the last phase because they are also part of the boundary tree and we can utilize them for faster query answering later on. \mycolor{The inner tree's time complexity is $O(|C||V_i| (max(|X_j|)\cdot n\cdot c^{2n-2}_{max}\log c_{max} (max(|X_j|)+ h_i)+\log |V_i|)$, where $|V_i|$ is the vertex set of subgraph $G_i$, $h_i$ is the inner tree height, and $X_j$ is the tree node belongs to $G_i$. As for the node number, tree width, tree height, and $c_{max}$ for most of the sub-graphs are much smaller, it is much faster than the original big tree. As the constructions are dependent, they can run in parallel to achieve higher performance. The total inner tree label size is $O(|C|\times |V_i|\cdot |h_i|\cdot c^{n-1}_{max})$.}  

\subsection{Boundary Tree}
\label{subsec:Forest_Boundary}
When we have a query from different partitions, we need inter-partition information to help answer it. One way is precomputing the all-pair boundary results. However, unlike the previous one, this all-pair result is much larger and harder to compute because it contains boundaries of all the partition and they are far away from each other. To avoid the expensive long-range \textit{Sky-Dijk}, we also resort to the concatenation-based \textit{MCSP-2Hop} to build a boundary tree.

\begin{figure}[htbp]
    \centering
    \includegraphics[width=2.9in]{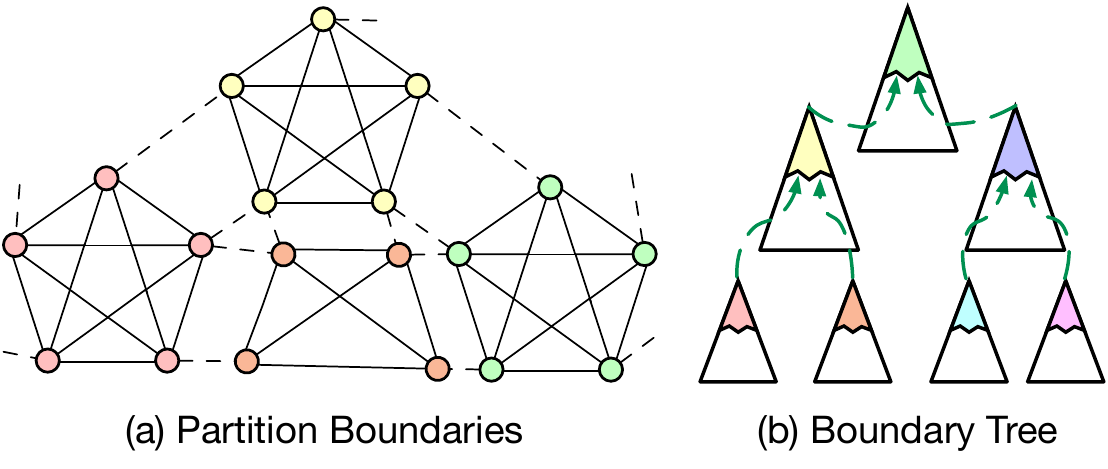}
    \caption{Forest Hop Labeling}
    \label{fig:Forest}
\end{figure}

The boundary tree is built on a graph that contains all the boundary edges and all the partition-boundary all-pair edges, as shown in Figure \ref{fig:Forest}-(a). Although we can still use the same degree elimination order to contract the vertices, its randomness ignores the property that these boundaries are the cuts of the partitions. As a result, the inner partition trees and the boundary tree has no relation and they can hardly interchange information. Therefore, we contract the boundaries in the unit of partition instead of degree: the boundaries from the same partition are contracted together in the same order of their inner tree. In this way, we can attach the inner partition trees to the boundary tree as shown in Figure \ref{fig:Forest}-(b). Furthermore, only the new shortcuts need skyline concatenation while the existing ones can be inherited directly.

In terms of label assignment, we can populate the labels only to the boundaries and we call it \textit{boundary label}. It has smaller index size but requires 6-hop in query answering. The boundary tree construction takes $O|B| (max(|X_j|)\cdot n\cdot c^{2n-2}_{max}\log c_{max} (max(|X_i|)+ h_B)+\log |B|)$ and its label size is $O(h_B\cdot |B|\cdot c^{n-1}_{max})$, where $B$ is the boundary set and $h_B$ is the boundary tree height. We can further push down the boundary labels to the inner labels and we call them \textit{extended labels}. Specifically, for each inner partition tree, its ancestors in the boundary tree are also added to its labels. For example in Figure \ref{fig:Forest}-(b), the green region have labels in every small tree, while the yellow region only appear in the yellow, red, and orange tree. The extended label assignment takes an extra $O(|C||V_i| (max(|X_j|\cdot c^{2n-2}_{max}\log c_{max} ( max(|X_j|)+h_i+h_B))+\log |V_i|))$ time and $O(h_B\cdot |V-B|\cdot c^{n-1}_{max})$ space.

\subsection{Query Processing}
\label{subsec:Forest_Query}
Depending on the locations of the two query points $s$ and $t$, we discuss the following situations:
\begin{enumerate}
	\item \textbf{$s$ and $t$ are in the same partition}: Because the all-pair boundary information is considered during construction, we can use the local partition tree directly to answer it.
	\item \textbf{$s$ and $t$ are all boundaries}: It can be answered with the boundary label directly.
	\item \textbf{$s$ is in one partition and $t$ is another's boundary}: If we only have the boundary label, then we need to first find the results from $s$ to its boundaries within the local label, then concatenate the results from these boundaries to $t$ using the boundary labels, and it takes 4 hops. If we have the extended labels, we can use it as a 2-hop label directly. 
	\item \textbf{$s$ and $t$ are in different partitions}: If we only have the boundary label, then we need to first find the results from $s$ to its boundaries within the local label, then concatenate the results from these boundaries to $t$'s boundaries using the boundary labels, and finally concatenate these results to $t$ with $t$'s local labels. Obviously, this is a 6-hop procedure. Similar to the previous case, if we have the extended labels, we can use it as a 2-hop label directly. More advanced query processing techniques will be presented in the next section.
\end{enumerate}

\section{Skyline Path Concatenation}
\label{sec:Skyline}
In this section, we present insights of the skyline path concatenation that speedup the index construction and query processing. Specifically, we first discuss how to concatenate two skyline path sets via single hop in Section \ref{subsec:Skyline_Full}, which is used in the \textit{tree node contraction} phase of index construction. Then we extend to multiple hops in \ref{subsec:Skyline_Multiple}, which is used in the \textit{label assignment} phase of index construction and query processing. Finally, we explain how to efficiently answer \textit{MCSP} in Section \ref{subsec:Skyline_Constraint} and how to apply these techniques to \textit{forest labels} in Section \ref{subsec:Skyline_Forest}.
\begin{figure}[ht]
    \centering
    \includegraphics[width=3.4in]{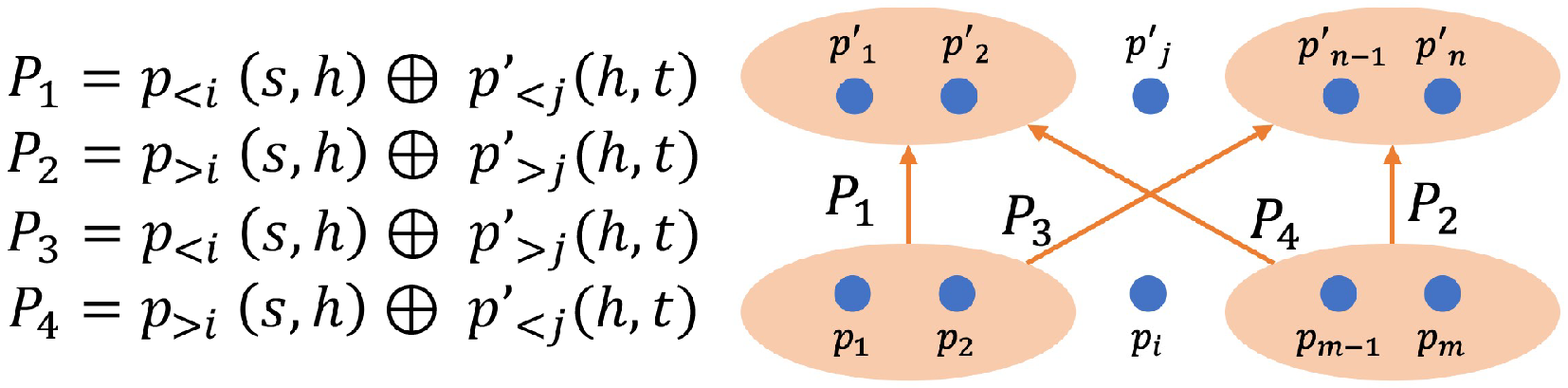}
    \caption{Concatenated Path Category}
    \label{fig:ConcatenationCategory}
\end{figure}

\subsection{Single Hop Skyline Path Concatenation}
\label{subsec:Skyline_Full}
Given two sets of skyline paths $P(s,h)$ and $P(h,t)$ with sizes of $m$ and $n$, it takes two operations ($O(mn)$ time \textit{path concatenation} and $O(mn\log mn)$ time \textit{skyline path validation}) to obtain the skyline result. Because the $mn$ \textit{concatenation} is inevitable, we aim to speedup the \textit{validation} by exploring the concatenation properties first. In the following, we first present the observations and algorithms for the 2D skyline path concatenation as it is simpler and has some unique characteristics. Then we generalize to the multi-dimensional case.

\subsubsection{2D Skyline Path Concatenation Observations}
The 2D skyline dataset has a unique characteristic compared with the multi-dimensional skyline dataset: when we sort all the data in the weight-increasing order, their costs are naturally sorted in the decreasing order. Therefore, we have the following four observations to help improve the concatenation efficiency. We first dig into the simplest \textit{1-n / m-1} skyline concatenation. The following property proves that we do not need validation for this simple case:
\begin{property}
\label{lemma:SingleConcatenation} \textit{\textbf{(1-n / m-1 Skyline)}}
$\forall p\in P(s,h)$, $p\oplus P(h,t)$ can generate $n=|P(h,t)|$ paths that cannot dominate each other. Similarly, $\forall p'\in P(h,t)$, $P(s,h)\oplus p'$ can generate $m=|P(s,h)|$ paths that cannot dominate each other.
\end{property}
\begin{proof}
We only need to prove the first half. Suppose we choose $p(w,c)$ from $P(s,h)$. Because $P(h,t)$ is the skyline path set, we select any two paths from it: $p_1(w_1,c_1)$ and $p_2(w_2,c_2)$. Suppose $w_1\leq w_2$ and $c_1 \geq c_2$, then this relation still holds because $w_1+w\leq w_2+w$ and $c_1+c \geq c_2+c$. 
\end{proof}


Next we explore the \textit{m-n} concatenation case. Suppose all the skyline paths are sorted in the weight-increasing order. We define $p_i(s,h)=(w_i,c_i)$ ($i\in[1,m]$) as the $i^{th}$ skyline path in $P(s,h)$, $p'_j(h,t)=(w'_j,c'_j)$ ($j\in[1,n]$) as the $j^{th}$ skyline path in $P(h,t)$, and $p_i^j$ as the concatenation path of $p_i(s,h)$ and $p'_j(h,t)$. Therefore, $p_1^1$ is the concatenation of the first skyline path in $P(s,h)$ and the first path in $P(h,t)$, and $P_m^n$ is the  concatenation of the last skyline path in $P(s,h)$ and the last one in $P(h,t)$. Then they have the following property:
\begin{property}
\label{lemma:FirstandLast} \textit{\textbf{(Endpoint Skyline)}}
$p_1^1$ and $p_m^n$ are two endpoints of the concatenation results and they are skyline paths.
\end{property}
\begin{proof}
$p_1^1=p_1(s,h)\oplus p'_1(h,t)=p(w_1+w'_1, c_1+c'_1)$. Because $w_1$ is the smallest in $p_1(s,h)$ and $w'_1$ is the smallest in $p'_1(h,t)$, $w_1+w'_1$ is still the smallest weight in $P_c$. Similarly, $p_m^n=p_m(s,h)\oplus p'_n(h,t)=p(w_m+w'_n, c_m+c'_n)$ has the smallest cost in $P_c$. Therefore, no other path in $P_c$ can dominate $p_1^1$ and $p_m^n$.
\end{proof}

Finally, we explore the domination relations among the paths between the endpoints. Given any concatenated path $p_i^j$, we can classify the other paths (except the ones generated by $p_i$ and $p'_j$) into four categorizes as shown in Figure \ref{fig:ConcatenationCategory}. 
The following two properties provide the dominance relations between $p_i^j$ and these four path sets.


\begin{property}
\label{theorem:ConcatenationUnDominance} \textit{\textbf{(Side-Path Irrelevant)}} The concatenated path $p_i^j$ cannot be dominated by the paths in $P_1$ and $P_2$.
\end{property}
\begin{proof}
$\forall p_k\in p_{<i}$ and $p'_k\in p'_{<j}$, $c_k > c_i$ and $c'_k>c_j$. Therefore, $c_k+c'_k > c_i+c'_j$, so $P_1$ cannot dominate $p_i^j$. Similarly, $\forall p_k\in p_{>i}$ and $p'_k\in p'_{>j}$, $w_k > w_i$ and $w'_k>w_j$. Therefore, $w_k+w'_k > w_i+w'_j$, so $P_2$ cannot dominate $p_i^j$.
\end{proof}

\begin{property}
\label{theorem:ConcatenationDominance} \textit{\textbf{(Cross-Path Relevant)}}
$p_i^j$ is a skyline path iff it is not dominated by any path in $P_3$ and $P_4$.
\end{property}
\begin{proof}
$\forall p_k\in p_{<i}$ and $p'_k\in p'_{>j}$, we have $c_k > c_i$ and $c'_k<c_j$, $w_k > w_i$ and $w'_k<w_j$, so $c_k+c'_k$ has no relation with $c_i+c'_j$, and $w_k+w'_k$ has no relation with $w_i+w'_j$. Therefore, paths in $P_3$ has the possibility to dominate $p_i^j$, and same for $P_4$. In other words, it is the paths in $P_3$ and $P_4$ that determines if $p_i^j$ could be a skyline path or not.
\end{proof}

\subsubsection{Efficient 2D Skyline Path Concatenation}
The above properties enable us to reduce the validation operations and speed up the 2D concatenation. The following theorem further reduces the comparison number in the validation operation:

\begin{theorem}
\label{theorem:GreedySkyline}
Suppose the concatenated paths are coming in the weight-increasing order. We have $k$ skyline paths and a new path $p_{k+1}$ comes with weight larger than $w_k$. If $p_k$ cannot dominate $p_{k+1}$, then $p_{k+1}$ must be a skyline path. 
\end{theorem}
\begin{proof}
$p_{k+1}$ cannot be dominated by latter paths because $w_{k+1}\leq w_j, \forall j>k+1$. $p_k$ has the smallest cost of the current skyline paths. If $c_k\geq c_{k+1}$, then the other paths' cost are also larger than $c_{k+1}$. Therefore, if $p_k$ cannot dominate $p_{k+1}$, then no  future path can dominate it.
\end{proof}



\begin{figure}[ht]
    \centering
    \includegraphics[width=2.3in]{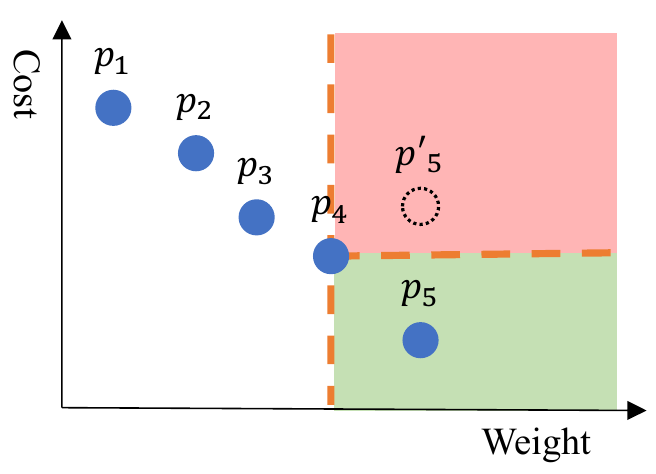}
    \caption{Next Skyline Path Validation}
    \label{fig:ConcatenationExample}
\end{figure}
For example in Figure \ref{fig:ConcatenationExample}, we have a set of skyline paths $\{p_1,p_2,p_3,p_4\}$ sorted in the weight-increasing order and $p_5$ is the next concatenated path with the next smallest weight. Therefore, $p_5$ can only locate in the right regions of $p_4$. If $p_5$ is in the red region, it is dominated by $p_4$ and discarded; If $p_5$ is in the green region, it will become the next skyline path. 


\begin{algorithm}[ht]
	\caption{Efficient 2D Skyline Path Concatenation}
	\label{Algorithm:GreedyConcatenation}
	\LinesNumbered
	\KwIn{Sorted Skyline Path Set $P(s,h)$ and $P'(h,t)$}
	\KwOut{Skyline Concatenated Path $P_s\subseteq P(s,h)\oplus P'(h,t)$}
	\SetKwProg{Fn}{Procedure}{}{}
	$p_1^1\leftarrow p_1(s,h)\oplus p'_1(h,t)$; $P_s.insert(p_1^1)$\;
	$p_1^2\leftarrow p_1(s,h)\oplus p'_2(h,t)$; $Q.insert(p_1^2)$\;
	$p_2^1\leftarrow p_2(s,h)\oplus p'_1(h,t)$; $Q.insert(p_2^1)$\;
	\While{$!Q.empty()$}
	{
	    $p_i^j\leftarrow Q.pop()$; $p_{max}\leftarrow P_s[|P_s|-1]$\;
	    \If{$c_i^j\leq c_{max}$}
	    {
	        $P_s.insert(p_i^j)$\;
	    }
	    $LazyInsert(i,j)$\;
	}
	\Return $P_s$\;
	\Fn{$LazyInsert(i,j)$}
	{
	    \If{$i+1<m$}
	    {
	        \If{$p_{i+1}^j$ is visited}
	        {
	            $LazyInsert(i+1,j)$\;
	            break\;
	        }
	    
	            $p_{i+1}^j\leftarrow p_{i+1}(s,h)\oplus p'_j(h,t)$\;
	            \If{$c_{i+1}^j>c_{max}$}
	            {
	                $LazyInsert(i+1,j)$\;
	            }
	            \Else{$Q.insert(p_{i+1}^j)$}
	        
	    }
	    \If{$j+1<n$}
	    { 
	        \If{$p_i^{j+1}$ is visited}
	        {
	            $LazyInsert(i,j+1)$\;
	            break\;
	        }
	            $p_i^{j+1}\leftarrow p_{i}(s,h)\oplus p'_{j+1}(h,t)$\;
	            \If{$c_{i}^{j+1}>c_{max}$}
	            {
	                $LazyInsert(i,j+1)$\;
	            }
	            \Else{$Q.insert(p_{i}^{j+1})$}
	        
	    }
	   }
\end{algorithm}

Based on Theorem \ref{theorem:GreedySkyline}, we propose an efficient skyline path concatenation as shown in Algorithm \ref{Algorithm:GreedyConcatenation}. As proved in Property \ref{lemma:FirstandLast}, $p_1^1=p_1(s,h)\oplus p'_1(h,t)$ is definitely a skyline result and we put it into the result set $P_s$. Because the next path with the smallest weight is either $p_1^2$ or $p_2^1$, we put them in a priority queue $Q$ as candidates. Then we iterate the candidate queue until it is empty. Each time we retrieve the top path $p_i^j$ in $Q$ and compare it with the last path in $P_s$. If $p_i^j$ is not dominated, we put $p_i^j$ into $P_s$. Then we generate the next two candidates $p_{i+1}^j$ and $p_i^{j+1}$ if they exist. However, we do not insert them into $Q$ directly as they could be dominated by the last path in $P_s$ already, or they could have already been generated by other paths. Therefore, we use the \textit{LazyInsert} to identify the next set of candidates that are not dominated (Line 16 and 25) and not generated (Line 12 and 21) recursively. In this way, the pruning power of the current skyline path is maximized, and we could have a $Q$ of smaller size and eliminate the redundancy.

\textit{Correctness:} Firstly, because we generate and test both $p_{i+1}^j$ and $p_i^{j+1}$ of all $p_i^j$ even if some of them are not inserted into $Q$, we have produced all $mn$ candidates. Secondly, because $Q$ is sorted by weight, and both $p_{i+1}^j$ and $p_i^{j+1}$ have larger weight than $p_i^j$, we visit the $p_i^j$ from $Q$ in the weight-increasing order. Therefore, according to Theorem \ref{theorem:GreedySkyline}, Algorithm \ref{Algorithm:GreedyConcatenation} returns a set of concatenated skyline paths.

\textit{Complexity:} The maximum depth of $Q$ is $log(mn)$, so it takes $O(mn\log(mn))$ to pop out all the elements. Since we just check the dominant relation once for each connecting path, the skyline validation time for each path is constant. Although the worst case time complexity of concatenating all the skyline paths via $h$ is still $O(mn\log(mn))$, the \textit{LazyInsert} pruning can keep the size of $Q$ small in practice.

\subsubsection{Multi-Dimensional Skyline Path Concatenation}
Although the previous observations, theorems, and algorithms efficiently reduce the complexity of skyline path validation in 2D road networks, they take the advantages of the 2D skyline path cost being natural ranking when sorted in weight-increasing order, which does not hold in the higher dimensional space. For instance, in algorithm \ref{Algorithm:GreedyConcatenation}, when a new path $p_i^j$ is concatenated, its weight must be no less than all the skyline paths in $P_s$. Thus, we only need to compare the cost of $p_i^j$ with the cost of the ranked last path in $P_s$, which has the minimum cost. However, this skyline validation method will be not feasible when each path has multiple costs, as we can not guarantee the ordering of costs for a weight-sorting multi-skyline path set. For example in Table \ref{table:EG_skylines}, it has five 4D skyline paths. Even though we sort them by the weigh-increasing order, the remaining three costs have no order at all. To reduce the comparison number of skyline validation in paths with multiple costs, we first investigate into the order properties of the multi-dimensional skyline paths:

\begin{definition}[Cost Rank]
	\label{def:Cost_Rank}
	Given a set of skyline paths in a $n$-dimensional graph, we sort them in each of the cost $c_i$ in the increasing order. Then $r_i[p_k]$ denotes the rank of path $p_k$ in cost $c_i$. 
	\end{definition}

For example, in Table \ref{table:EG_skylines}, we have a set of skyline paths $\{p_1, p_2, p_3, p_4, p_5\}$. On each dimension, we rank the paths with increasing order, which is shown in Figure \ref{fig:MultiSkylineValidation} (ignore $p_6$ at this stage). Then $r_1[p_4]=1$ because $p_4$ has the smallest value in cost $c_1$. Similarly, $r_2[p_5]=2$ and $r_3[p_4]=3$. For each path $p$, it has a rank in each dimension, and the following \textit{distinct criterion} is the one we are most interested in:

\begin{definition}[Distinct Criterion]
	Among the $n$ rankings of a path $p$ in $G$, its \textit{Distinct Criterion} is the one with the highest rank denoted as $c_{dc}[p]$.
\end{definition}

\begin{table}[]
\centering
\caption{An Example of 4D Skyline Paths}
\label{table:EG_skylines}
\begin{tabular}{|l|l|l|l|l|}
\hline
$P_c$ & $w$ & $c_1$ & $c_2$ & $c_3$ \\ \hline
$p_1$ & 2   & 4     & 10    & 2     \\ \hline
$p_2$ & 2   & 4     & 9     & 3     \\ \hline
$p_3$ & 3   & 7     & 7     & 10    \\ \hline
$p_4$ & 5   & 3     & 9     & 3     \\ \hline
$p_5$ & 6   & 6     & 8     & 4     \\ \hline
\end{tabular}
\end{table}

For example, $p_3$'s distance criterion is $c_2$ as it is the smallest compared its ranking in other criteria: $r_w[p_3]=3$, $r_1[p_3]=5$, and $r_3[p_3]=5$. Similarly, the distinct criterion of $p_1$ is $w$, $c_{dc}[p_2]$ is $c_3$, $c_{dc}[p_4]$ is $c_1$, and $c_{dc}[p_5]$ is $c_2$. Then we utilize the distinct criterion to further investigate the dominant relation in the high-dimensional space:

\begin{lemma}
\label{lemma:Multi-Non-Dominate}
	Suppose we already have $k$ skyline paths, and we aim to test a new path $p_{k+1}$ with weight no-smaller than the existing ones. After inserting into the existing criteria sorting lists, we can obtain the ranks $m_i$ of $p_{k+1}$ in each criteria $c_i$. Suppose $p_{k+1}$ is ranked last among all the paths with the same value. Then all the paths $P_i^{>m}$ that have ranking lower than $m_i$ have larger value than $p_{k+1}$ in criteria $c_i$, and they cannot dominate $p_{k+1}$.
\end{lemma} 
\begin{proof}
	Because the dominance relation requires smaller value, they cannot dominate $p_{k+1}$.
\end{proof}

For example in Figure \ref{fig:MultiSkylineValidation}, we insert a new path $p_6=(8,5,8,9)$. As $p_6$ has greater weight than the five existing paths, it must be ranked last on dimension $w$. Then, we find the ranks of $p_6$ on other dimensions, which are highlighted in red. $p_6$ is ranked 4, 3 and 5 in $c_1$, $c_2$ and $c_3$. Then along criteria $c_2$, paths $P_c^{>3}=\{p_2,p_4,p_1\}$ have lower ranking and larger values than $p_4$, so they cannot dominate $p_6$. We can have the following lemma by taking complementary set of Lemma \ref{lemma:Multi-Non-Dominate}.

\begin{lemma}
	\label{lemma:Multi-Dominate}
	Following Lemma \ref{lemma:Multi-Non-Dominate}, all the paths $P_i^{\leq m}$ that have ranking higher in criteria $c_1$ than $p_{k+1}$ are possible to dominate $p_{k+1}$.
\end{lemma}
2
For example in Figure \ref{fig:MultiSkylineValidation}, $P_1^{\leq 4}=\{p_4,p_1,p_2\}$ can possibly dominate because they rank higher than $p_6$ along $c_1$.  Now we have the sufficient and necessary requirement for $p_{k+1}$ being a new skyline path:
\begin{theorem}
	\label{theorem:MultiGreedySkyline}
	Following Lemma \ref{lemma:Multi-Dominate}, $p_{k+1}$ is a skyline path if and only if none of the $m-1$ higher ranking paths $P_i^{\leq m}$ in any criteria $c_i$ can dominate $p_{k+1}$.
\end{theorem}

\begin{proof}
	Because the paths $P_i^{\leq m}$ have either equal or smaller values than $p_{k+1}$ along criteria $c_i$, then the other paths $P_i^{>m}$ have larger values and cannot dominate $p_{k+1}$. Therefore, if none of the paths in $P_i^{\leq m}$ can dominate $p_{k+1}$, then no existing skyline path can dominate $p_{k+1}$, so $p_{k+1}$ is a skyline path.
\end{proof}

\begin{figure}[ht]
    \centering
    \includegraphics[scale=0.45]{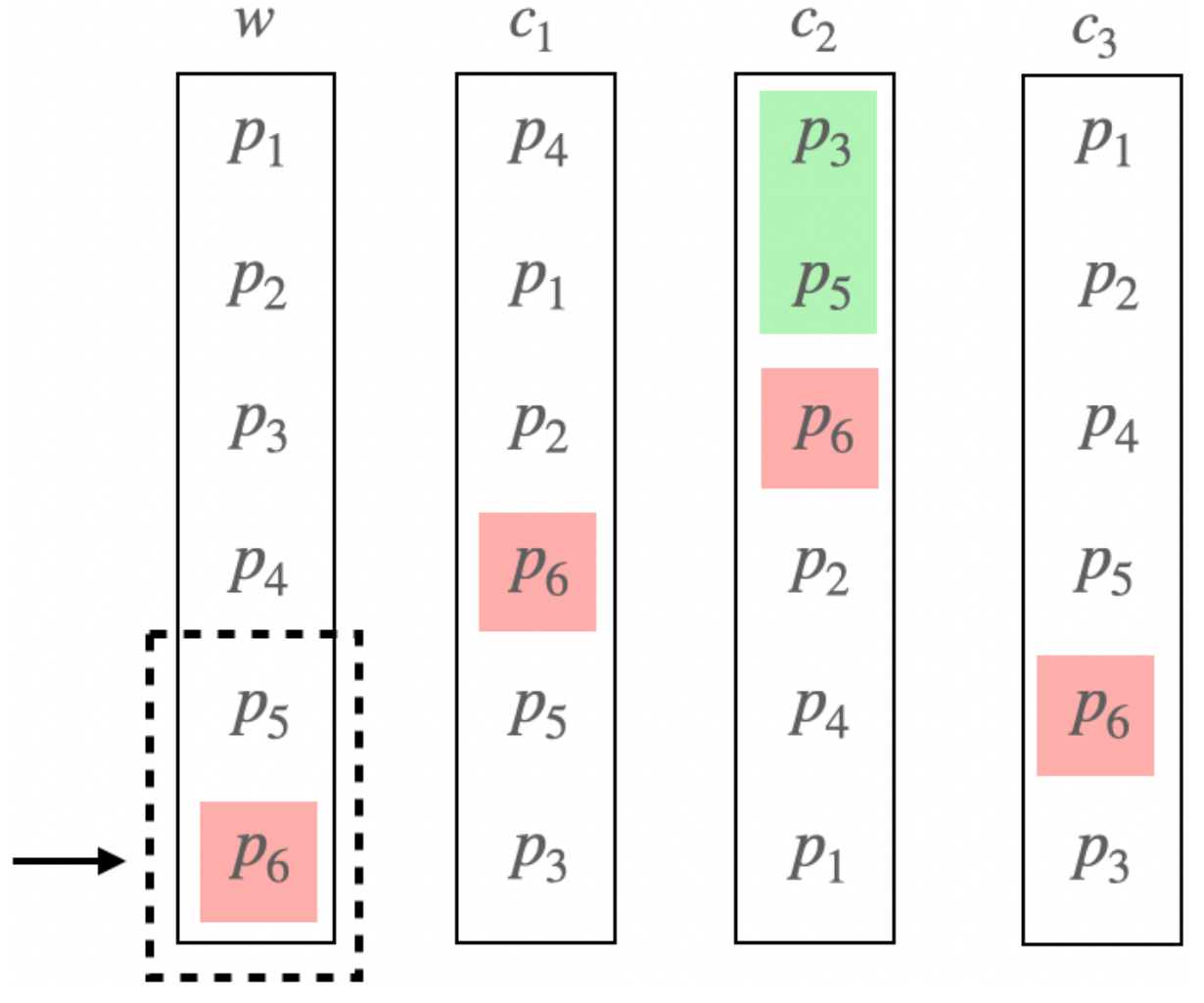}
    \caption{Next Skyline Path Validation in Multi-dimensions. New $p_6=(8,5,8,9)$.}
    \label{fig:MultiSkylineValidation}
\end{figure}


Although any criterion can work for Theorem \ref{theorem:MultiGreedySkyline}, not all of them have the same quality. For instance, along the dimension $w$, $p_{k+1}$ is ranked last and has $k$ possible dominating paths, so we have to compare $p_{k+1}$ with all $k$ previous paths. Therefore, using the criterion with the least number of possible dominating path would be beneficial, and the distinct criterion $c_{dc}[p]$ is guaranteed to have the smallest $P_i^{\leq m}$ as $p_{k+1}$ ranks highest along it. Then we have the following tighter condition:

\begin{theorem}
	\label{theorem:MultiGreedySkylineDC}
	$p_{k+1}$ is a skyline path if and only if none of the $m-1$ higher ranking paths $P_i^{\leq m}$ in the distinct criteria $c_{dc}[p]$ can dominate $p_{k+1}$.
\end{theorem}


The complexity of using Theorem \ref{theorem:MultiGreedySkylineDC} is $O((|P_{r_{dc}[p_{k+1}]}^{\leq m}|+log k)\times n)$. It requires obtaining all the ranks first and comparing with all the criteria. One way to improve the efficiency is converting the value comparison to the dominating path set intersection:

\begin{theorem}
	\label{theorem:MultiGreedySkylineIntersection}
	Suppose we have obtained all the $P_i^{\leq m_i}$, then	$p_{k+1}$ is a skyline path $\Leftrightarrow\bigcap_{i\in[1,n]} P_i^{\leq m_i} = \phi$
\end{theorem}
\begin{proof}
	Suppose $\exists p_i \in \bigcap_{i\in[1,n]} P_i^{\leq m_i}$, then $p_i$ has higher rank than $p_{k+1}$ in all criteria, so $p_i$ dominates $p_{k+1}$ and $p_{k+1}$ is not a skyline path.
\end{proof}

As the intersection result cannot become larger, we can compute the result in the $|P_i^{\leq m_i}|$ increasing order to keep reducing the intermediate result size. This order can also be used in the value testing method. To further reduce the sorting complexity, we apply the \textit{lazy comparison} to only continue the sorting on the criteria when needed. We first introduce the hierarchical approximate rank as follows:

\begin{definition}[Hierarchical Approximate Rank]
	\label{def:HA-Rank}
	Given a set of $m$ sorted numbers $S$ and a query value $q$, $q$'s rank binary search at level $h_0$ is $HA_0(q)=m/4$ if $q<S[m/2]$ and $HA_0(q)=m/2$ if $q\geq S[m/2]$. Suppose we are at level $h_i$ and the current \textit{HA-Rank} is $HA_i(q)$, then for the next level $h_{i+1}$, we have the following three scenarios: 1) $HA_{i+1}(q)=(2HA_i(q)-1)/2$ if $q<S[HA_i(q)]$; 2) $HA_{i+1}(q)=HA_i(q)$ if $q\geq S[HA_i(q)]$ and $q< S[HA_{i-1}(q)]$; and 3) $HA_{i+1}(q)=(2HA_i(q)+1)/2$ if $q\geq S[HA_i(q)]$ and $q\geq S[HA_{i-1}(q)]$. Regardless of the level number, the \textit{HA-Rank} are compared with their values.
\end{definition}

The intuition behind it is that we set the \textit{HA-Rank} optimistically. When $q$ is smaller than the value at the current  \textit{HA-Rank}, we set it to a smaller \textit{HA-Rank}; when $q$ is not smaller than the value at the current \textit{HA-Rank} but smaller than the previous one, we keep the current \textit{HA-Rank}; when $q$ is not smaller than the values at the current \textit{HA-Rank} and the previous one, we set it to a larger rank.

Now in our case of $k$ $n$-dimensional skyline paths, we sort them using the \textit{HA-Rank} with query value $p_{k+1}$. Specifically, each time we update the \textit{HA-Rank} of the criterion that has the highest \textit{HA-Rank} value until one criterion has reached the actual data. Then this criterion is $p_{k+1}$'s distinct criterion, and we have also got $p_{k+1}$'s rank on it. We can keep sorting to get the $2^{nd}$, $3^{rd}$ till the last distinct criterion if keeps searching on. Algorithm \ref{Algorithm:Multi-Validation} shows the details of the validation using the lazy comparison with \textit{HA-Rank}.

\begin{algorithm}[ht]
	\caption{Multi-Dimension Skyline Path Intersection Validation}
	\label{Algorithm:Multi-Validation}
	\LinesNumbered
	\KwIn{$n$-Dimensional Skyline Paths $\{p_1,\dots,p_k\}$, new path $p_{k+1}$}
	\KwOut{If $p_{k+1}$ is a skyline path}
	$(c, m)\leftarrow HA\_Sort(p_{k+1})$; \Comment{$c$ is the distinct criteria and $m$ is its corresponding rank}\;
	$P'\leftarrow P^{\leq m}_c$\;
	$counter\leftarrow 1$\;
	\While{$P'\neq \phi$ and $counter< n$}
	{
		$(c', m')\leftarrow HA\_Sort(p_{k+1})$\;
		$P' \leftarrow P' \bigcap P^{\leq m'}_{c'}$\;
		$counter\leftarrow counter+1$\;
	}
	
	\If{$P'\neq \phi$}
	{
		\Return{$p_{k+1}$ is a skyline path}
	}
	\Else
	{
		\Return{$p_{k+1}$ is not a skyline path}
	}

\end{algorithm}

However, the efficiency of Algorithm \ref{Algorithm:Multi-Validation} depends on the number of the paths having higher or equal ranking than the new coming path $p_{k+1}$. If $p_{k+1}$ ranks very low on distinct criterion, then we may compare a large amount of skyline paths, which is still time consuming. Thus, we provide a pruning technique to identify the non-skyline earlier with the help of the following theorem.

\begin{theorem}
	\label{theorem: MultiGreedyPruning}
	Suppose we have obtained all the $P_i^{> m_i}$, then $p_{k+1}$ is not a skyline path $\Leftrightarrow |\bigcup_{i\in[1,n]} P_i^{\leq m_i}| < k$.
\end{theorem}

\begin{proof}
	 We proof it by contraction. Suppose $p_{k+1}$ is a skyline path and there are $k-1$ skyline paths having lower ranking than $p_{k+1}$ on at least one criterion. As $p_{k+1}$ has the larger weight than these $k$ paths. We only need to consider the rankings on the criteria of costs. Except $p_{k+1}$, $k$ paths are considered in the rankings. Thereby, it must has at least one skyline path $p_i$, $\forall i \leq k$, having higher or equal ranking than $p_{k+1}$ on all the criteria of costs. According to definition \ref{def:Multi-Dominance}, $p_i$ can dominate $p_{k+1}$. Since the assumption is false, Theorem \ref{theorem: MultiGreedyPruning} is true.
\end{proof}

Therefore, when $c_{dc}[p_{k+1}]$ is very low like falling into the last quarter, we can use Theorem \ref{theorem: MultiGreedyPruning} to do the validation instead of using the intersection. We can further terminate earlier when the distinct criteria is smaller than a threshold according to the following theorem:

\begin{theorem}
	\label{theorem: MultiGreedyPruning2}
	$p_{k+1}$ is not a skyline path $\Leftrightarrow r_{dc}[p_{k+1}] > k + 1 - \dfrac{k}{n-1}=k(1-\dfrac{1}{n-1})$.
\end{theorem}
\begin{proof}
	Because $r_{dc}[p_{k+1}]$ has the highest rank, then $p_{k+1}$'s ranks on other criteria cannot be higher. In other words, there are at most $k+1-r_{dc}[p_{k+1}]$ paths in each $P_i^{> m_i}$. Therefore, if $(k+1-r_{dc}[p_{k+1}])\times (n-1)< k$, then $\sum |P_i^{> m_i}|<k\Rightarrow |\bigcup_{i\in[1,n]} P_i^{\leq m_i}| < k$. 
\end{proof}

Although Theorem \ref{theorem: MultiGreedyPruning2} seems loose at the first glance, it becomes power when $k$ is big and $n$ is small. For instance, when $n$ is 4 and $k$ is 10,000, then when $r_{dc}[p_{k+1}]$ is larger than $10000\times(1-\dfrac{1}{4-1}=6667)$, we can prune $p_{k+1}$ safely.

\subsection{Multiple Hop Skyline Path Concatenation}
\label{subsec:Skyline_Multiple}

In the previous subsections, we have illustrated how to identify all the skyline paths from $s$ to $t$ when $h_i$ is the unique intermediate hop. However, the intersection of $s$' out-label and $t$'s in-label may contain multiple hops. In this case, the identified skyline paths via $h_i$ could be dominated by the skyline paths via other intermediate hops. Therefore, after we produce all the skyline paths via each intermediate hop, we need further validation on these path sets. Obviously, it is very time-consuming due to the large concatenation number and time complexity. Therefore, in this section, we introduce some pruning methods to screen the skyline paths via different hops before the actual concatenation to reduce total concatenation number. We first introduce the \textit{Rectangle Pruning} for \textit{2}-dimensional graph in Section \ref{subsubsec:Skyline_Multiple_Rectangle}, and then generalize it to $n$-Cube pruning for \textit{n}-dimensional graph in Section \ref{subsubsec:Skyline_Multiple_Cube}. 

\begin{figure}[ht]
    \centering
    \includegraphics[width=2.3in]{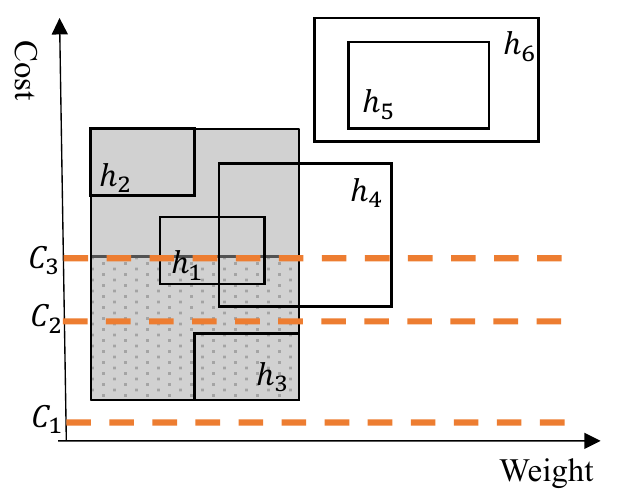}
    \caption{Multiple Skyline Rectangles and \textit{CSP} Query Rectangle (Shaded)}
    \label{fig:ConcatenationPruning}
\end{figure}

\subsubsection{Rectangle Pruning for 2D Graph}
\label{subsubsec:Skyline_Multiple_Rectangle}
Given two skyline path sets $P(s,h)$ and $P(h,t)$, according to Property \ref{lemma:FirstandLast}, $p_1^1$ and $p_m^n$ are the two end points and they can determine a rectangle $[w_1^1,w_m^n]\times [c_1^1, c_m^n]$  ($p_1^1$ as the top-left and $p_m^n$ as the bottom-right) such that all the other skyline paths must locate in it. Therefore, we can use this rectangle to approximate the skyline paths distribution of a hop without computing them. In other word, this rectangle is a \textit{MBR (Minimum Bounding Rectangle)} for all the skyline paths via $h$. For example in Figure \ref{fig:ConcatenationPruning}, suppose we have the skyline rectangles of six hops and each of them is identified with only two computations for the two endpoints. Among these end points, we find one with the smallest weight and another one with the smallest cost, then we can obtain a new \textit{hop rectangle} that guarantees the all the skyline results must be inside it (the shaded area formed by top-left of $h_2$ and the bottom-right of $h_3$). We only consider the rectangles of hops that intersect or inside \textit{query rectangle} as our final concatenated hops. For instance, as the rectangle $h_5$ and $h_6$ are outside the query rectangle, we can prune them directly. For $h_4$, we further consider the paths with weight less than $w(p_n.h_3)$. Finally, all the paths in $h_1$, $h_2$ and $h_3$ need actual concatenations. The detailed pruning procedure is described in Algorithm \ref{Algorithm:MultiHop}. The time complexity is $O(max|X_i|)$ as $max|X_i|$ is max size of $H$. 

For example, in the label assignment example of Section \ref{subsubsec::CSP2Hop_TD_Label}, we first compute three rectangles of the corresponding hops: $h_1:[6,11]\times [10,5]$; $h_5:[4,7]\times [6,5]$; $h_8:(2,3)$. Because $(2,3)$ dominates the other two rectangles, we prune $h_1$ and $h_5$ and return $(2,3)$ directly as the skyline path.



\begin{algorithm}[ht]
	\caption{Rectangle-based Hop Pruning}
	\label{Algorithm:MultiHop}
	\LinesNumbered
	\KwIn{The Hop Set $H$ in \textit{LCA} of $s$ and $t$}
	\KwOut{Pruned hop set $H'$}
	$TL\leftarrow (\infty,0), BR\leftarrow(0,\infty)$;\algorithmiccomment{Two end-point}\\
	\ForEach{$h_i\in H$}{
	 $P(s,h_i) \leftarrow L(s)[h_i], P(h_i,t) \leftarrow L(t)[h_i]$\;
	 
	 $p_1^1 \leftarrow p_1(s,h_i) \oplus p_1(h_i,t)$;\algorithmiccomment{Top Left}\\
	 $p_m^n \gets p_m(h_i,t) \oplus p_n(h_i,t)$;\algorithmiccomment{Bottom Right}\\
	 \If{$p_1^1.w< TL.w$}
	 {
	    $TL	\leftarrow p_1^1$\;
	 }
	 \If{$p_m^n).c< BR.c$}
	 {
	    $BR\leftarrow p_m^n$\;
	 }

	 $Rectangle[h_i]\gets(p_1^1, p_m^n)$\;
	}

	$R\leftarrow (TL, BR)$\;
	\ForEach{$h_i\in H$}{
	    \If{$R$  $\cap Rectangle[h_i] \neq \phi$}{
	    $H'.insert(h_i)$
	    }
	}
	\Return $H'$\;
\end{algorithm}


\subsubsection{$n$-Cube Pruning for n-Graph}
\label{subsubsec:Skyline_Multiple_Cube}





Different from the 2D scenario, in the multi-dimensional road networks, such rectangles can not be found, as the Property \ref{lemma:FirstandLast} cannot hold for the concatenated skyline paths with over 2 criteria. This is caused by the fact that the path with the largest weight does not necessarily have the largest cost in each dimension. For instance in Table \ref{table:EG_skylines}, $p_5$ has the largest $w$ but it does not have the largest cost in any of the three cost dimensions. Therefore, we can only find a hop rectangle in each 2D subspace, i.e., hop rectangles based on all two criteria combinations. By intersecting the projection of these faces, we obtain a hypercube that encloses all the concatenated skyline paths. For better description of the pruning in $n$-graph, we first map the $n$-graph to a $n$-dimensional space: Given a $n$-dimensional graph $G(V,E)$ and $n$-dimensional space $D^n=\{d_0,d_1,...,d_{n-1}\}$, for arbitrary path $p$, we use dimension $d_0$ to describe the value of $w(p)$, and $d_i$ to describe the value of $c_i(p)$, where $i \in [1, n-1]$. Now we are ready to describe the $n$-Cubes.



\begin{definition}[Skyline Path Set Intervals]
\label{def:Skyline_Intervals}
	Given a path set $P_c$ and its corresponding space $D^n$, $\forall d_i\in D^n$, its lower-bound $Min(d_i)$ is the smallest value of all path $p\in P_c$ on dimension $d_i$. Similarly, its upper-bound $Max(d_i)$ is the largest value of all path on $d_i$. An interval $I_i=[Min(d_i), Max(d_i)]$ covers all the path values in $P_c$ on dimension $d_i$.
\end{definition}

\begin{definition}[$n$-Cube]
\label{def:n_cube}
 Given a path set $P_c$, a $n$-cube $\Gamma^n(P_c)=\{I_0,I_1,..,I_{n-1}\}$ covers all the skyline paths in $P_c$.
\end{definition}

We use the following \textit{supremum} and \textit{infimum} to assist the pruning:
\begin{definition}[$n$-Cube Supremum and Infimum]
For a $n$-cube $\Gamma^n$, its supremum $Sup[\Gamma^n(P_c)]$=$\{Max(d_0),$ $...,Max(d_{n-1})\}$ is a set that contains the upper-bounds of $P_c$ on each dimension, and its infimum $Inf[\Gamma^n(P_c)]=\{Min(d_0),...,Min(d_{n-1})\}$ is a set that contains the lower-bounds of $P_c$.
\end{definition}

For instance in Table $\ref{table:EG_skylines}$, there are five skyline paths in $P_c$. We use $d_0$ to denote skyline path values on criterion $w$, and $d_1$ to $d_3$ to denote the values on the criteria from $c_1$ to $c_3$. As the smallest value on $d_0$ is 2 and the largest on is 6, we can find a set $I_0=[2,6]$ to contain the five skyline path values on $d_0$. By the same token, we obtain $I_1=[3,7]$, $I_2=[7,10]$, and $I_3=[2,10]$. Therefore, we can identify a 4-cube $\Gamma^4(P_c)=\{[2,6], [3,7], [7,10], [2,10]\}$ for the example skyline paths. Note that this 4-cube can cover all skyline path values in table $\ref{table:EG_skylines}$. In addition, $Inf[\Gamma^4(P_c)]=\{2,3,7,2\}$ and $Sup[\Gamma^4(P_c)]=\{6,7,10,10\}$. Naturally, we have the following property:

\begin{property}
\label{property:n-cube}
    Given a set of $n$-dimensional skyline paths $P_c$, its $n$-cube $\Gamma^n(P_c)$ contains all the paths in $P_c$.
\end{property}

Similar to the 2D scenario, we first create the $n$-cube for each hop to approximate the concatenation result space. According to definition \ref{def:n_cube}, to identify a $n$-cube of hop $h_k$, we need to computing both the lowest bound and the highest bound on each dimension. Specifically , we compute the concatenated skyline path $p_{Min(d_i)}$ with the smallest value on dimension $d_i$. Suppose $p$ has the smallest value on $d_i$ in $P(s,h_k)$ and $p'$ has the smallest value on $d_i$ in $P(t,h_k)$. Then $Min(d_i)$ can be obtained by concatenating $p$ and $p'$. If $p_{Min(d_i)}$ is a unique path that has the smallest value on $d_i$, we can easily prove that $p_{Min(d_i)}$ must be a concatenated skyline path. Thus, the $d_i$ value of $p_{Min(d_i)}$ must be the lower-bound of $d_i$ for all the concatenated skyline paths. We can use the same way to compute the upper-bound by finding $Max(d_i)$ in both $P(s,h_k)$ and $P(t,h_k)$. However, different to the 2D case, we can not guarantee that the concatenated path with the largest value on $d_i$ is a $n$-dimensional skyline path. 
Therefore, the obtained upper-bounds are actually looser than the actual bounds. For ease of presentation, we refer to $n$-cube of concatenated skyline path set via $h_k$ as $\Gamma(h_k)$, and $\Gamma(h_k)\supseteq \Gamma(P(s,h_k)\oplus P(h_k,t))$. Also, we refer to $Sup[\Gamma^n(h_k)]$ and $Inf[\Gamma^n(h_k)]$ as $Sup(h_k)$ and $Inf(h_k)$, if no special instructions on dimensions are needed. 


\begin{figure}[ht]
    \centering
    \includegraphics[scale=0.3]{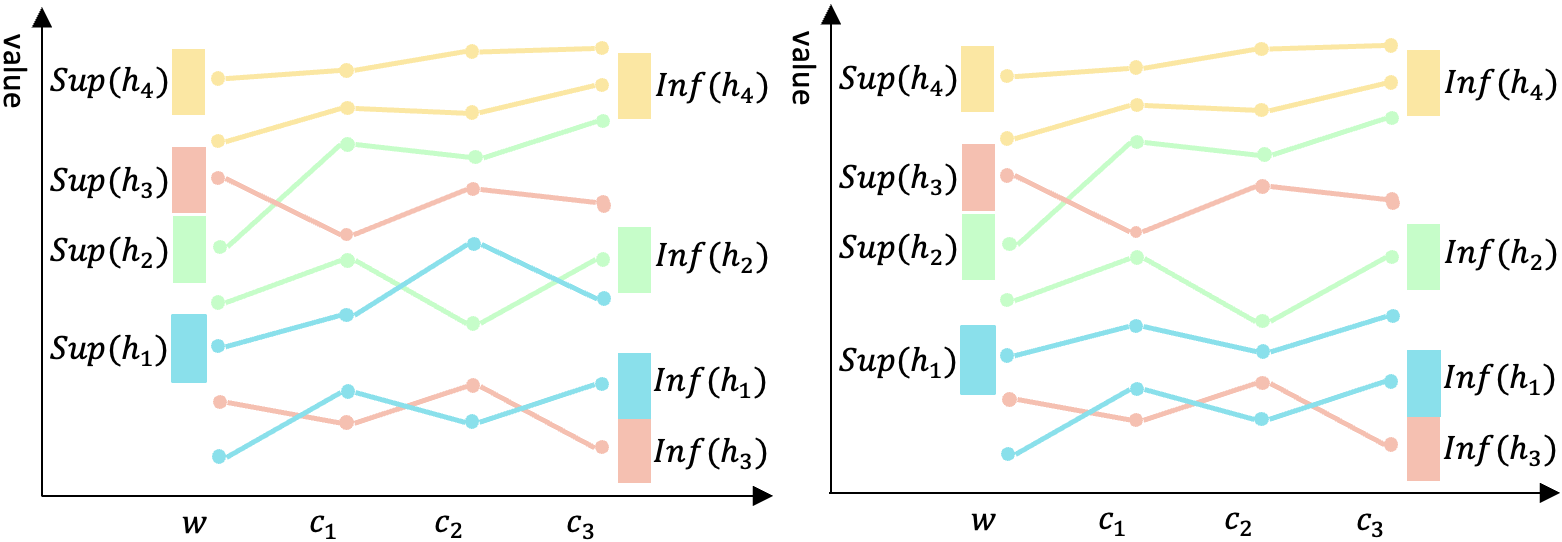}
    \caption{Supremum and Infimum for Hops (left) and Supremum Shrinking( right)}
    \label{fig:SubAndInf}
\end{figure}

Because the $n$-cube is a hypercube, it is difficult to clearly demonstrate a range of skyline paths in an example plot like as the hop rectangle in 2D. Nevertheless, if we treat the dimensions as $x$-coordinates and supremum/infimum values on each dimension as $y$-coordinates, we are able to map a $n$-cube to a 2D plot in the form of two polylines. For instance, Figure \ref{fig:SubAndInf} (left) shows four $4$-cubes, each represented by two polylines with the same color. In the following, we discuss the $n$-cube pruning with the help of Figure \ref{fig:SubAndInf}.

\begin{property}
	\label{property:n-cube_dominance}
	\textbf{($n$-Cube Dominance)} Given any two $n$-cube $\Gamma_i$ and $\Gamma_j$, if $Sup[\Gamma_i]$ dominates $Inf[\Gamma_j]$, then $\Gamma_i$ dominates $\Gamma_j$ and all the paths in $\Gamma_j$ are also dominated by the paths in $\Gamma_i$ .
\end{property}

For example in Figure \ref{fig:SubAndInf}, the polyline $Sup(h_4)$ is over the supremum of other three $4$-cube. Therefore, all the paths in $\Gamma(h_4)$ have greater value on each dimension than the paths in other $4$-cubes. However, as the polylines of other three $4$-cubes are intersected with each other, we need further pruning method on them.

As discussed previously, the infimum of each estimated hop's $n$-cube $\Gamma(h_i)$ is fixed, but the supremum is loose. Therefore, it is promising to further prune more hops by shrinking the supremums, which could be achieved by concatenating the skyline paths on some hops earlier. For instance in Figure \ref{fig:SubAndInf}, if we first compute all the concatenated skyline paths on $h_1$, we can obtain a fixed $Sup(h_1)$, which is shrunken to the position in Figure \ref{fig:SubAndInf} (right). Thus, $\Gamma(h_2)$ can be pruned directly by $\Gamma(h_1)$. Thereby, it is important to determine an efficient concatenating order for the hops, which is defined as follows.

\begin{definition}[$n$-Cube Comparison]
	\label{def:n-Cube_Comparions}
	Given two of $n$-Cubes $\Gamma_i$ and $\Gamma_j$ that cannot dominate each other. We use $|Inf^{<}[\Gamma_i]|$ to denote the number of dimensions $Inf[\Gamma_i]$ is smaller, and $|Inf^<[\Gamma_j]|$ to denote $\Gamma_j$'s. Then we say $\Gamma_i < \Gamma_j$ if $|Inf^<[\Gamma_i]|>|Inf^<[\Gamma_j]|$. The draw of comparison can be settled by comparing the supremum. The further draw is settled by comparing infimum dimensions in the decreasing order, and the one with first smaller dimension lower-bound is smaller. 
\end{definition}

The intuition behind this order is that if a hop has lower infimum, it has higher change contain skyline paths. Especially, the hop with the lowest infimum must contain skyline paths. Moreover, a $n$-cube with a lower infimum has higher chance to prune more $n$-cubes after it shrinks. Based on this observation, we have the following order the $n$-cubes.

\begin{definition}[$n$-Cube Concatenation Priority]
	\label{def:Priority}
	Given a set of $n$-Cubes that cannot dominate each other, their concatenation order is the same as their increasing infimum order.
\end{definition}


For instance, because all the lowest bounds of $Inf(h_1)$ are lower than $Inf(h_2)$, we need to concatenate skyline paths on $h_1$ before $h_2$. $\Gamma_1<\Gamma_3$ because $|Inf^{<}[\Gamma_1]|=|Inf^{<}[\Gamma_3]|$ but $|Sup^{<}[\Gamma_1]|>|Sup^{<}[\Gamma_3]|$. Finally, $\Gamma_3<\Gamma_2$ because $|Inf^{<}[\Gamma_3]|<|Inf^{<}[\Gamma_2]|$. Therefore, the three $4$-Cubes should be concatenated in the order of $\Gamma_1$, $\Gamma_3$, and $\Gamma_2$. It should be noted that this order can be change dynamically when any $n$-cube is pruned or when any supremum is fixed. The detailed procedure is shown in Algorithm \ref{Algorithm:n-cube_pruning}. 


\begin{algorithm}[ht]
	\caption{$n$-Cube Hop Pruning}
	\label{Algorithm:n-cube_pruning}
	\LinesNumbered
	\KwIn{The Hop Set $H$ in \textit{LCA} of $s$ and $t$}
	\KwOut{Skyline Paths Set $P_{s,t}$}
	$P_{s,t}\gets \phi$\;
	\ForEach{$h_i\in H$}
	{
		\ForEach{$d_j\in D^n$}
		{
			$Inf[\Gamma_i][j]\gets p_1(s,h_i).d_j+p_1'(h_i,t).d_j$\;
			$Sup[\Gamma_i][j]\gets p_n(s,h_i).d_j+p_m'(h_i,t).d_j$\;
		}
		$\Gamma_i\gets (Inf[\Gamma_i], Sup[\Gamma_i])$\;
	}
	$\hat{\Gamma}\gets $n-Cube-Sort($\{\Gamma_i\})$\;
	$Sup[\hat{\Gamma}]\gets [-1]\times n$\;
	\ForEach{$\Gamma_i\in \hat{\Gamma}$}
	{
		\If{$Sup[\hat{\Gamma}]$ dominates $Inf[\Gamma_i]$}
		{
			$Continue$\;
		}
		$P_{s,t}\gets Skyline(P_{s,t}, P(s,h_i)\oplus P(h_i,t))$\;
		\ForEach{$d_j\in D^n$}
		{
			\If{$Sup[\hat{\Gamma}< Max(P_{s,t}.d_j)$}
			{
				$Sup[\hat{\Gamma}\gets Max(P_{s,t}.d_j)$\;
			}
		}
	}
	\Return $P_{s,t}$\;
\end{algorithm}

Finally, the rectangle pruning in the 2D scenario can be viewed as a special case of the $n$-Cube pruning with $n=2$. Although its supremum is fixed initially, its performance can also be improved with the 2-Cube concatenation priority.



\begin{figure*}[ht]
    \centering
    \includegraphics[width=6.8in]{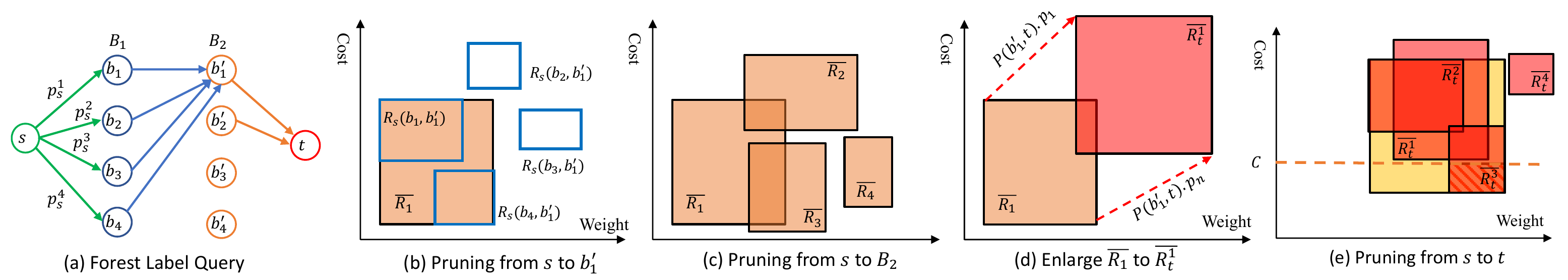}
    \caption{Constraint Shortest Path Query with Forest Labels}
    \label{fig:ForestQuery}
\end{figure*}

\subsection{Constraint Skyline Path Concatenation}
\label{subsec:Skyline_Constraint}
In this section, we further dig into the concatenation when the query constraints $\overline{C}$ are provided. Similar to the previous sections, we also first discuss the easier 2D case and then generalize it to the $n$ dimensional case.
\subsubsection{CSP Concatenation}
\label{subsec:Skyline_Constraint_Single}
We first deal with the \textit{CSP} query $q(s,t,C)$. As illustrated in Figure \ref{fig:ConcatenationPruning}, the query rectangle is further reduced to a smaller one by applying an upper-bound of $C$ (like the dotted area created by $C_3$). If $C$ is smaller than the lowest cost (like $C_1$), then the result is empty. If there exists one rectangle whose upper-left point is the only left most point in the query rectangle (like $C_2$ and $h_3$), we return this point directly. Otherwise, there must be some partial rectangles. We first find the most upper-left point in the query region and denote it as $p^*$ and use for the current best point ($w^*$ as its weight). Then for all the partial rectangles whose left boundary is larger than $w^*$, we can prune them directly (e.g., $h_4$ is pruned by $h_3$). After that, we enumerate the candidates in the remaining partial rectangles (like $h_1$) in a similar way like 
Algorithm \ref{Algorithm:GreedyConcatenation}. The difference is we only start testing when the sub-path's cost is smaller than $C$, and stop if either the current minimum weight is larger than $w^*$ or we obtain the first one whose cost is no larger than $C$. Traversal from multiple partial hops can run in parallel as they have the same termination condition and do not interfere each other. 


\subsubsection{MCSP Concatenation}
\label{subsec:Skyline_Constraint_Multi}
Different from the $n$-Cube in Section \ref{subsubsec:Skyline_Multiple_Cube}, the supremums that are larger than the constraint $\overline{C}$ can be replaced by  $\overline{C}$ during query answering. Besides, a $n$-Cube can be further pruned by the following property.

\begin{property}
	\label{Property:MCSP_Constraint}
	Given a $n$-Cube and a set of query constraints $\overline{C}$, $n$-Cube can be pruned if any of its infimum is larger than the corresponding cost.
\end{property} 

Then for the concatenation of each $n$-Cube, because the paths are concatenated in the distance-increasing order, we only need to concatenate the first one that satisfies $\overline{C}$ for each hop. Within each hop concatenation, the paths in the labels are also pruned when any of its cost does not satisfy $\overline{C}$. When have the first candidate, its distance can also be used to prune the remaining concatenations for other hops.

\subsection{Forest Constraint Skyline Path Concatenation}
\label{subsec:Skyline_Forest}
\subsubsection{Rectangle Pruning for 2D Graph}
In this section, we present how to apply the above techniques to answer \textit{MCSP} with forest labels. The extended labels can use the methods in Section \ref{subsec:Skyline_Constraint} directly. As for the forest labels, we break it into three parts as shown in \mycolor{Figure \ref{fig:ForestQuery}}-(a): 1) from $s$ to its boundary $B_1$, 2) from $B_1$ to $B_2$, and 3) from $B_2$ to $t$. However, different from the previous query that needs concatenation, the skyline results of 1) and 3) can be retrieved from their sub-trees directly because their boundaries are all their ancestors. As for the boundaries results, we first obtain their rectangles, then obtain the approximate results from $s$ to $B_2$ by viewing each boundary in $B_1$ as a hop. For example in \mycolor{Figure \ref{fig:ForestQuery}}-(b), we get a candidate rectangle $\bar{R_1}$ from $s$ to $b'_1$ by pruning the multi-hop of $B_1$. By applying the same procedure to the boundaries $B_2$ in parallel, we can have a set of candidate rectangles as shown in \mycolor{Figure \ref{fig:ForestQuery}}-(c). After that, these rectangles will be ``enlarged and moved" to the new locations by viewing the first and last skyline paths of each $B_2$ to $t$ as a vector as shown in \mycolor{Figure \ref{fig:ForestQuery}}-(d). Then we have a final set of rectangles $\bar{R_t}$ that approximate the skyline path space from $s$ to $t$. Finally, we can apply the constraint $C$ to shrink the search to a small area (shaded area in \mycolor{Figure \ref{fig:ForestQuery}}-(e)) and conduct the actual concatenation similar as Section \ref{subsec:Skyline_Constraint}.

\subsubsection{n-Cube Pruning for n-Graph}

Similar to the 2D scenario, the query processing is the same as the steps in Figure \ref{fig:ForestQuery}-(a). However, the rectangle pruning will be displaced by $n$-cube pruning. When pruning from $s$ to $b_1'$, we construct an $n$-cube for each hop in $B_1$. According to Property \ref{property:n-cube_dominance}, we can prune the dominated $n$-cubes.Then, we approximate the concatenation result space from $s$ to $b_1'$ by updating an infimum and a supremum depending on the undominated $n$-cubes in $B_1$. If it is the case, we compare their infimums to find the lowest value on each dimension and compare their supremums to find the highest values. For example in Figure \ref{fig:ForestQuery}-(a), assume we are considering 4-cube. For the four hops in $B_1$, we have $\Gamma_1=\{[4,7],[5,7],[5,8],[1,9]\}$, $\Gamma_2=\{[3,6],[4,6],[4,7],[3,10]\}$,
$\Gamma_3=\{[8,12],[8,10],
[9,12],[10,13]\}$, and 
$\Gamma_4=\{[7,12],$\\$[8, 10],[9,14],[12,14]\}$
, $\Gamma_3$ and $\Gamma_4$ can be pruned in terms of Property \ref{property:n-cube_dominance}. By viewing the infimums and supermums of $\Gamma_1$ and $\Gamma_2$,
a new infimum from $s$ to $b_1'$ is $\{3,4,4,1\}$ and a new supremum is $\{7,7,8,10\}$. Afterwards, we enlarge the $n$-cube from $s$ to $b_1'$ by adding on the infimum and the supremum of the $n$-cube from $B_2$ to $t$. For instance, assume the $n$-cube from $b_1'$ to $t$ is $\Gamma_1'=\{[3,5],[2,6],[5,9],[1,4]\}$, we can get the concatenation result space from $s$ to $b_1'$ to $t$ as $n$-cube, which has the infimum $\{6,6,9,2\}$ and the supremum $\{12,13,17,14\}$. We can use the same way to compute the $n$-cubes from $s$ to $t$ via other hops in $B_2$. We prune the dominated $n$-cubes from $s$ to $t$ and update a infimum and a supremum as the concatenation result space from $s$ to $t$ by viewing the remaining $n$-cubes. We can apply the constraints $\overline{C}$ as the final supremum to shrink the $n$-cube from $s$ to $t$.

\section{Experiment}
\label{sec:Experiment}

\subsection{Experimental Settings}
\label{sec:Experimental Settings}
We implement all the algorithms in C++ with full optimizations. The experiments are conducted in a 64-bit Ubuntu 18.04.3 LTS with two 8-cores Intel Xeon CPU E5-2690 2.9GHz and 186GB RAM.

    \begin{figure*}[ht]
    	\centering
    	\includegraphics[width=6.8in]{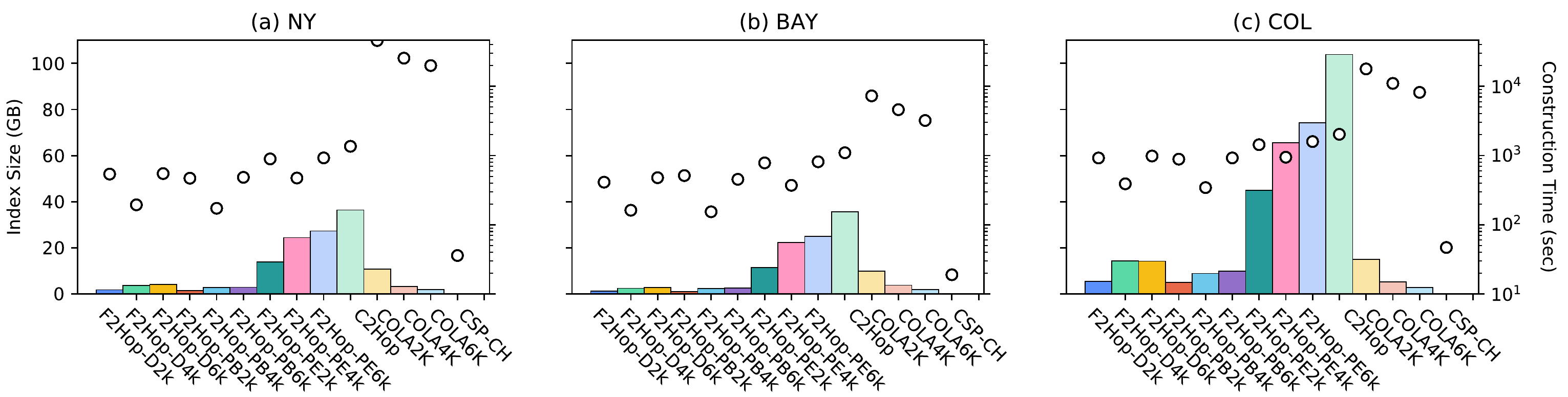}
    	\caption{Index Size (Bar) and Construction Time (Ball) of Different Partition Sizes, Boundary Contraction Orders, and Algorithms in 2-Graph}
    	\label{fig:Exp_Size}
    \end{figure*}

\textbf{Datasets}. 
All the tests are conducted in three real road networks obtained from the $9^{th}$ DIMACS Implementation Challenge \footnote{http://users.diag.uniroma1.it/challenge9/download.shtml}. 
Specifically, \textit{NY} is a dense grid-like urban area with 264,246 vertices, 733,846 edge, and several partitions connected by bridges. \textit{BAY} has a shape of doughnut around the \textit{San Francisco Bay} with 321,270 vertices, 800,172 edges, and several bridges connecting the banks. \textit{COL} has an uneven vertex distribution (dense around Denver but sparse elsewhere) with 435,666 vertices and 1,057,0666 edges. We use the road length as the weight $w$ and travel time as the first cost to test the performance of \textit{CSP}. We also generate two cost sets that have random and negative correlation with the distance for each network, and generate three more positive correlation costs for \textit{MCSP} tests. 



\textbf{Algorithms}. We extend the following state-of-the-art exact \textit{CSP} algorithms to their \textit{MCSP} versions \cite{DBLP:conf/adc/ZhangLZL21}: \textit{Sky-Dij} \cite{hansen1980bicriterion}, \textit{cKSP} \cite{gao2010fast}, 
\textit{eKSP} \cite{SEDENONODA2015602}, 
\textit{CSP-CH} \cite{storandt2012route}, \textit{COLA} \cite{wang2016effective} with $\alpha$=$1$. We denote our methods as \textit{C2Hop} and \textit{F2Hop}. The straightforward version (as implemented in \cite{Liu2021Forest}) is denoted as \textit{F2Hop-S}.

\textbf{Query Set}. 
For each dataset, we randomly generate five sets of OD pairs $Q_1$ to $Q_5$ and each has 1000 queries. Specifically, we first estimate the diameter $d_{max}$ of each network 
\cite{meyer2003delta}. Then each $Q_i$ represents a category of OD pairs falling into the distance range $[d_{max} / 2^{6-i}, d_{max}/2^{5-i}]$. For example, $Q_1$ stores the OD pairs with distances in range $[d_{max}/32, d_{max}/ 16]$ and $Q_5$ are the OD pairs within $[d_{max}/2, d_{max}]$. 
Then we assign the query constraints $C$ to these OD pairs. Firstly, we compute the cost range [$C_{min}, C_{max}$] for each OD. This is because if $C$\textless $C_{min}$, the optimal result can not be found then the query is invalid; and if $C$\textgreater$C_{max}$, the optimal results is the path with the shortest distance. 
Furthermore, in order to test the constraint influence on the query performance, we generate five constraints for each OD pair: $C=r\times C_{max}+(1-r)\times C_{min}$, with $r=0.1, 0.3, 0.5, 0.7$ and 0.9. When $r$ is small, $C$ is closer to the minimum cost; and when $r$ is big, $C$ is closer to the maximum cost. This constraint generation method is applied on all the cost dimensions.


\subsection{Index Size and Construction Time}
	\label{sec:Index Size and Construction Time}
	In this section, we first test the influence of the partition sizes and boundary contraction orders on index size and construction time. We use the 2-Graph as for the test and select the best ones for the higher dimensional tests. Specifically, we test the \textit{F2Hop}s and \textit{COLA} with three different partition sizes for each graph ($2k, 4k$, and $6k$), \textit{C2Hop} and \textit{CSP-CH} with no partition. In terms of contraction order, \textit{F2Hop-D} uses the degree elimination order with boundary label, \textit{F2Hop-PB} uses the partition order with boundary label, and \textit{F2Hop-PE} uses the partition order with the extended label. As shown in Figure \ref{fig:Exp_Size}, \textit{F2Hop} all have smaller index sizes and shorter construction time than \textit{C2Hop} because of the smaller tree sizes and parallel construction. Specifically among the \textit{F2Hop}s, \textit{F2Hop-PB} is the fastest and smallest \textit{F2Hop} because it has the smallest number of labels to construct, and \textit{F2Hop-PE} is the largest and slowest because it has the largest number of labels. It also shows that the partition order performs better than the degree elimination order in the boundary tree.  Among the baseline methods, \textit{COLA} also has a small index but it takes the longest time to construct because it uses \textit{Sky-Dijk} search. \textit{CSP-CH} takes the shortest time to construct the smallest index. In the following tests, we use $4k$ as the default partition and \textit{PB} as the boundary order and boundary label.
		
	Next, we test the index size and construction time in the higher dimensional graphs and the results are shown in Figure \ref{fig:Exp_SizeMulti}. Because \textit{COLA} already takes more than $10^4$ seconds in the 2-dimensional graph, we do not test it in the higher dimensional graphs. As for the \textit{F2Hop} and \textit{F2Hop-S}, they construct the same index so their index sizes are the same and we show them together as \textit{F2Hop Size}. As for the construction time, because they use different strategies so we distinguish them. Specifically, \textit{F2Hop} is faster than \textit{F2Hop-S} by less than $10\times$. \textit{CSP-CH} is faster in 2-dimensional but the construction time soars up in 3-dimensional, and takes longer than $10^5$ in higher dimensions so we do not show \textit{CSP-CH}'s performance on 4 and 5-graphs. This is caused by the explosion of the skyline path numbers on the entire graph, while our \textit{F2Hop} methods reduce it with forest. As for the construction time \textit{F2Hop}s, 3-graph takes an order of magnitude longer time than 2-graph, and 4-graph takes another order of magnitude longer time than 3-graph. This phenomenon demonstrates the influence of dimension number of costs on the \textit{MCSP} problem. However, compared with the 2 to 3 and 3 to 4, 5-graph's construction time is longer but is still on the same order of 4-graph. This is because the increase of skyline path slows down in higher dimensions. 
		
	\begin{figure}[h]
        \centering
    	\includegraphics[width=3.3in]{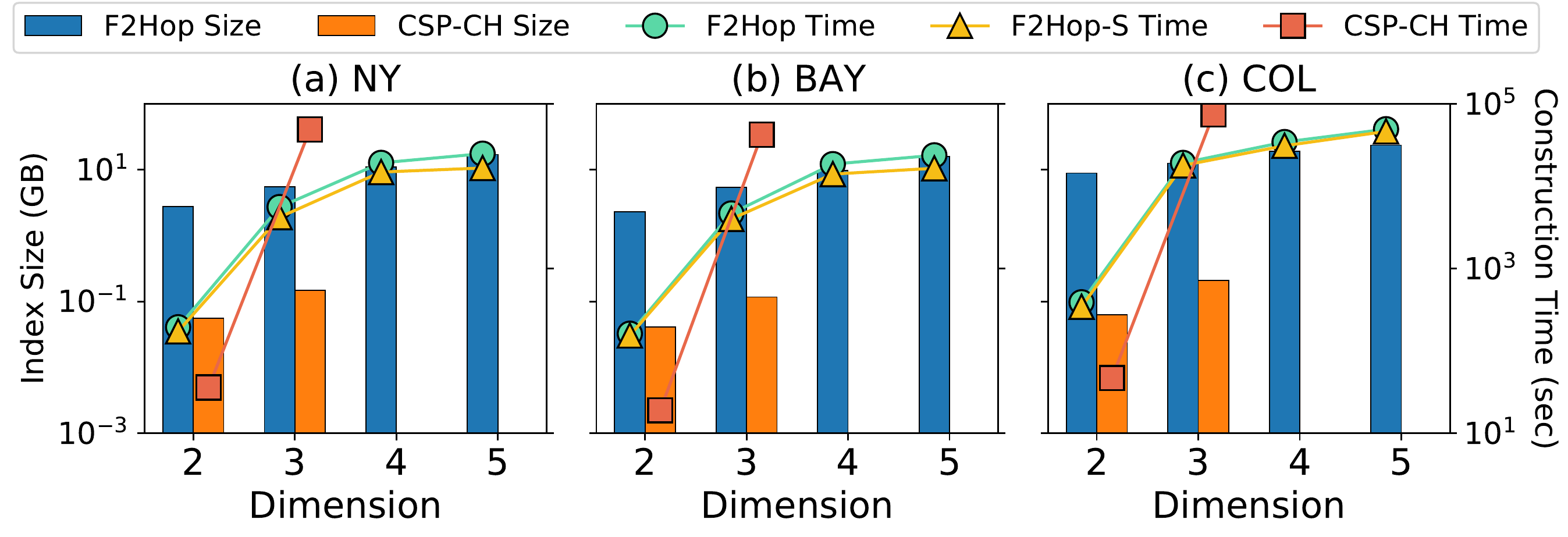}
    	\caption{Index Size (Bar) and Construction Time (Ball) of $n$-Graphs}
    	\label{fig:Exp_SizeMulti}
    \end{figure}




	\begin{figure}[ht]
        \centering
    	\includegraphics[width=3.2in]{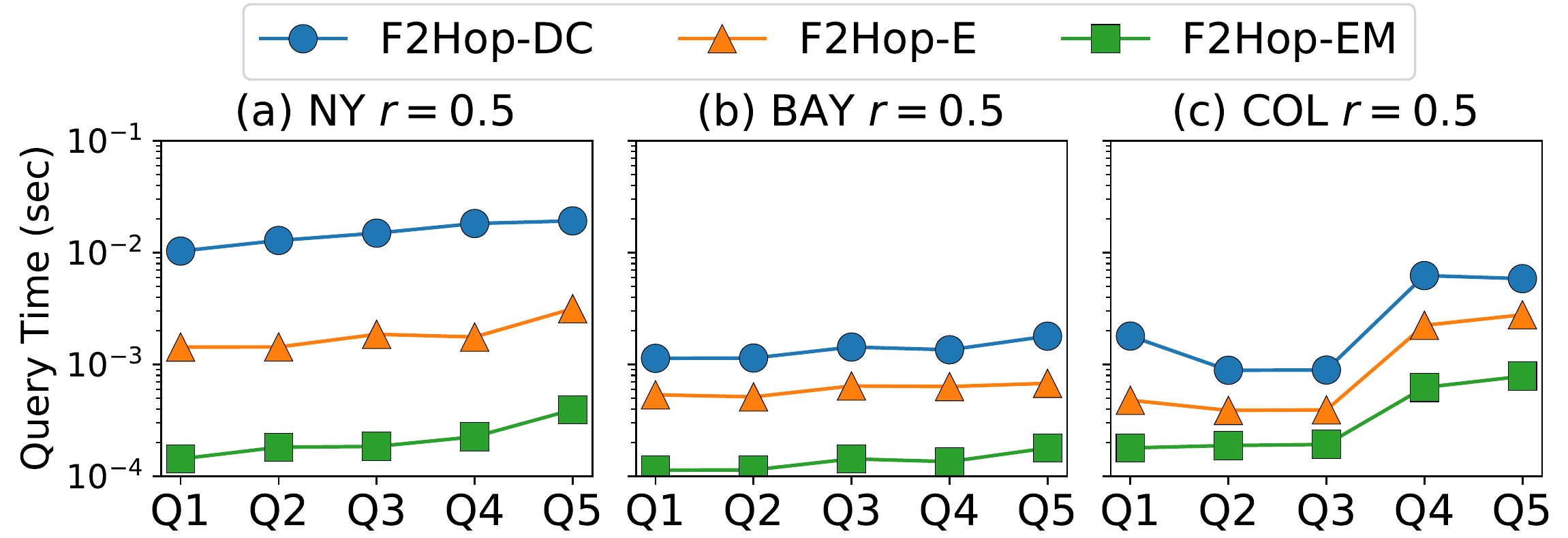}
    	\caption{Effectiveness of Skyline Concatenation of CSP}
    	\label{fig:Exp_Concatenation}
    \end{figure}
	
	\subsection{Effectiveness of Skyline Concatenation}
	In this section, we study the effectiveness of our proposed skyline concatenation methods: \textit{F2Hop-DC} uses the classic \textit{Divide-and-Conquer}, \textit{F2Hop-E} uses our \textit{Efficient concatenation} for \textit{CSP}, and \textit{F2Hop-EM} further applies the \textit{Multi-hop pruning} based on \textit{F2Hop-E}. 
	As shown in Figure \ref{fig:Exp_Concatenation}, \textit{F2Hop-E} is faster than \textit{F2Hop-DC}, while \textit{F2Hop-EM} is the fastest. This validates the effectiveness of our concatenation methods. Therefore, we only use \textit{F2Hop-EM} as the concatenation algorithm in the latter experiments. Among the networks, although \textit{NY} is relatively small, it is denser than the other two so the query performance is worse. \textit{BAY} is always faster because it has smaller label size thanks to its ring-shape topology. \textit{COL}'s $Q_1$ is worse than its $Q_2$ and $Q_3$ because the shorter queries have higher chance to fall into the dense Denver city and the mid-range query have higher chance to fall into the mountains and canyons.  
	
	\begin{figure}[ht]
        \centering
    	\includegraphics[width=3.3in]{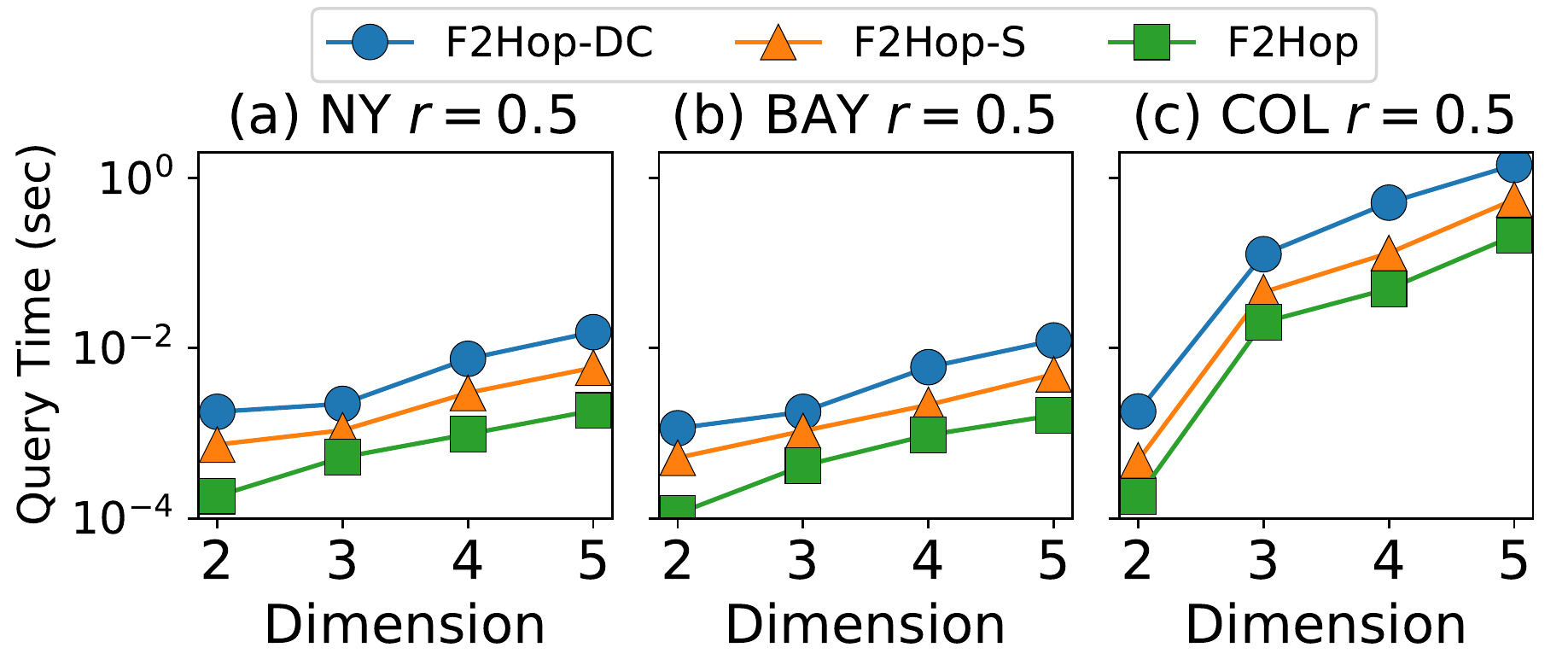}
    	\caption{Effectiveness of Skyline Concatenation in MCSP}
    	\label{fig:Exp_Concatenation_MCSP}
    \end{figure}

	For the higher dimensional cases, we further test our \textit{MCSP}  concatenation techniques and compare with the techniques for \textit{CSP}. The query sets are $Q_3$ with $r=0.5$. As shown in Figure \ref{fig:Exp_Concatenation_MCSP}, we can further achieve a nearly one order of magnitude faster concatenation speed.

	\subsection{Query Performance}
	\label{subsec:Experiment_QueryEfficiency} 
    \subsubsection{Query Distance}
    Figure \ref{fig:Exp_Query} (a)-(c) show the query processing time of different algorithms in distance-increasing order with $r=0.5$. Firstly, all \textit{CSP-2Hop} methods outperform the others by several orders of magnitude. Among them, \textit{C2Hop} is the fastest that runs in $10^{-5}$ seconds because it has the largest label size and only needs one multi-hop concatenation. \textit{F2Hop-PE} is only slightly slower because it also has a large index size. \textit{F2Hop-PB} is slower running in $10^{-3}$ seconds but it has the smallest index size and is the fastest to construct. Secondly, all the algorithms' running time increase as the distance increases. Compared to the other methods, the \textit{CSP-2Hop} methods are more robust as their running time only increase slightly. Among the comparing methods, when the ODs are near to each other, like $Q_1$ and $Q_2$, all of them can finish in 10 seconds. When the distance further increases, some of them soar up to thousands of seconds. Specifically, $cKSP$ is always the slowest one because it uses the pruned \textit{Dijkstra} search to generate and test the $k$ shortest paths in the distance-increasing order. Therefore, when the OD is far away from each other, there could enumerate a huge number of possible paths. The second slowest is \textit{Sky-Dij}, which has a huge amount of intermediate skyline paths to expand before reaching the destination. \textit{eKSP} is slightly faster than them because it avoids the expensive \textit{Dijkstra} search and uses the path concatenation instead. \textit{CSP-CH} is faster as expected, but its performance also deteriorates dramatically in $Q_4$ and $Q_5$. Finally, \textit{COLA} also tends to be stable because its query time is bounded by the inner-partition search, and inter-boundary results are precomputed regardless of the actual distance. 
    In summary, our \textit{CSP-2Hop} is orders of magnitude faster and also robust in terms of different distances.  

    \begin{figure}[ht]
    	\centering
    	\includegraphics[width=3.3in]{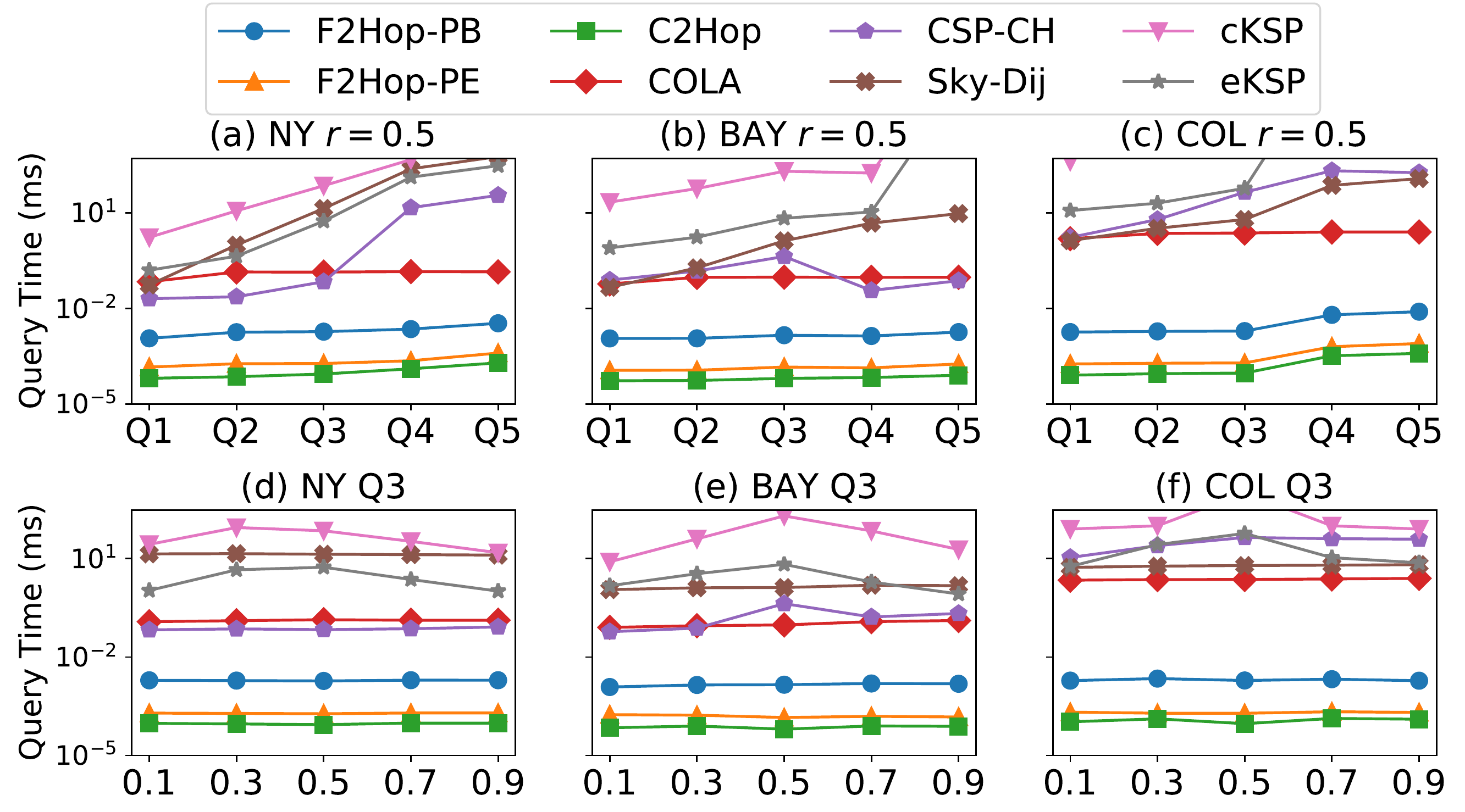}
    	\caption{CSP Query Performance of Distance and Ratio}
    	\label{fig:Exp_Query}
    \end{figure} 

%

    \subsubsection{Constraint Ratio}
    Figure \ref{fig:Exp_Query} (d)-(f) show the influence of the constraint ratio $r$, 
    and there are several interesting trends. Firstly, \textit{CSP-2Hop} is still robust and fastest under different $r$. Secondly, the baseline methods fall into two trends: \textit{Sky-Dij} and \textit{CSP-CH} that grows as the $r$ increases, while \textit{cKSP} and \textit{eKSP} have a hump shape. It should be noted that because the time is shown in logarithm, a very small change over $10^2$ in the plot actually has a very large difference. 
    The reason of the increase trend is that \textit{Sky-Dij} and \textit{CSP-CH} expand the search space incrementally until they reach the destination, during which they accumulate a set of partial skyline paths. Therefore, when the $r$ is small, partial skyline path set is also small because most of the infeasible ones are dominated. When $r$ is big, the partial skyline path set becomes larger so it runs slower. As for the hump trend, the two \textit{KSP} algorithms test the complete paths iteratively in the distance-increasing order. Therefore, when $r$ is big, more paths would satisfy this constraint, so there is higher chance to enumerate only a few paths to find the result. In fact, they are even faster than \textit{CSP-CH} in $Q_4$ and $Q_5$ when $r=0.9$ (not shown here). On the other hand, when $r$ is small, large amounts of infeasible paths are pruned during the enumeration so the running time is also short. But when $r$ is not big or small, the pruning power becomes weak and the actual result is much longer than the shortest path, they have to enumerate tens of thousands of path to find it, so query performance deteriorates. 

	\subsubsection{\mycolor{MCSP Query}}
	We test the \textit{MCSP} queries with $Q_3$ and $r=0.5$. As shown in Figure \ref{fig:Exp_MCSP}, \textit{F2Hop} are still orders of magnitude faster than the baseline. The performance of \textit{eKSP} looks stable because it is an enumeration-based method and we terminate the enumeration when more than $10k$ paths have been tested.  Its current performance is at the cost of failing answer some queries and bounded by $10k$ enumeration. 

	\begin{figure}[ht]
    	\centering
    	\includegraphics[width=3.3in]{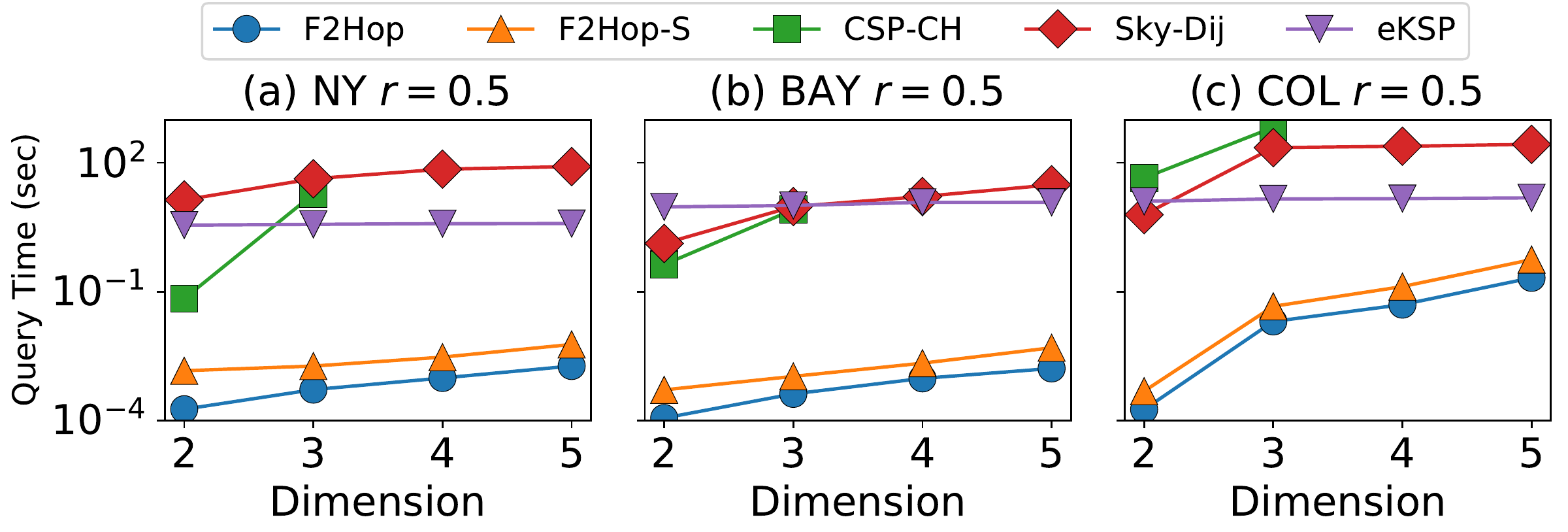}
    	\caption{MCSP Query Performance}
    	\label{fig:Exp_MCSP}
    \end{figure} 

	\subsubsection{Cost Correlation}
	We use the 2-dimensional index size and construction time to demonstrate the influence of the correlation. As shown in Figure \ref{fig:Exp_Correlation}, the positive correlation has the smallest index size and smallest construction time, and the negative correlation has the largest index size and longest construction time. This phenomenon is also caused by the \textit{skyline} itself instead of our index structure. When the criteria have negative correlations, the skyline number tends to increase to the worst square case. As a result, the query time increases dramatically for the baselines, while ours only increases slightly.
	
	\begin{figure}[ht]
    	\centering
    	\includegraphics[width=3.3in]{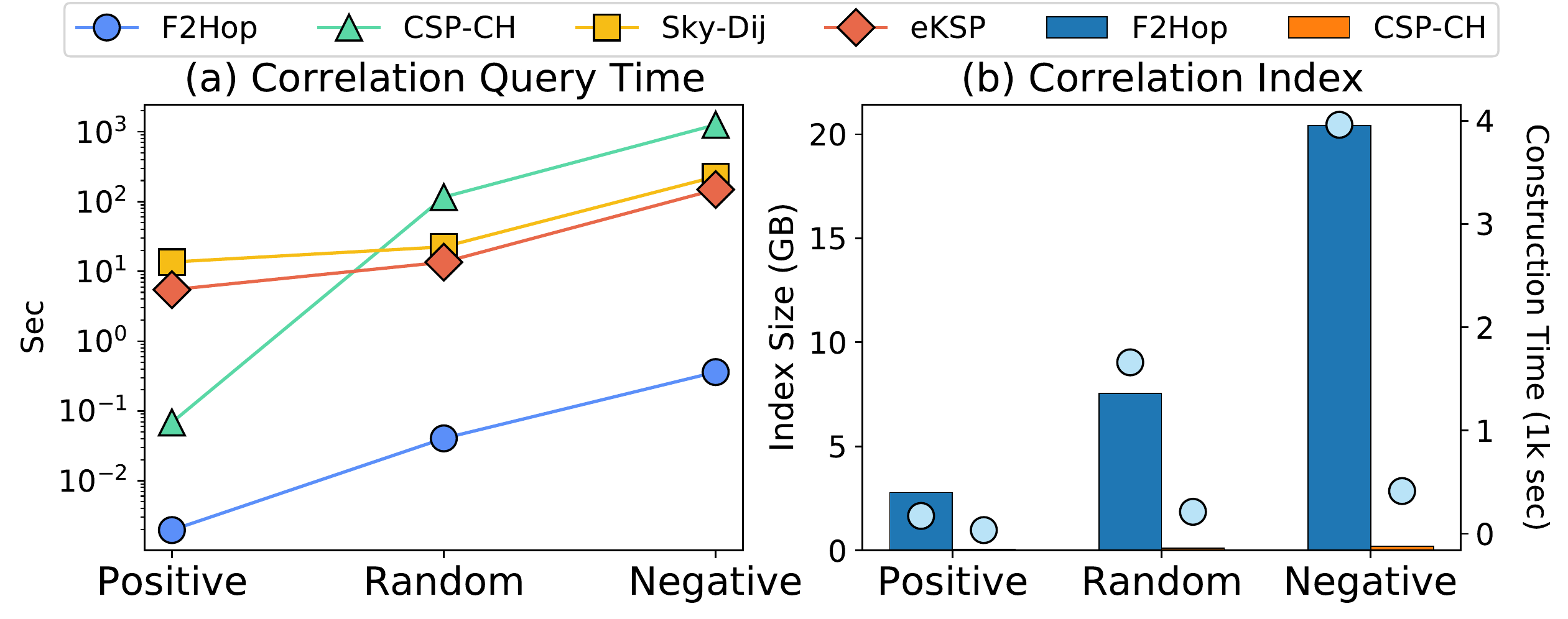}
    	\caption{Performance under Different Correlations (Bar: Size; Ball: Time)}
    	\label{fig:Exp_Correlation}
    \end{figure} 

\section{Related Work}
\label{sec:RelatedWork}

\subsection{Shortest Path and Skyline Path}
\label{subsec:RelatedWork_SkylinePath}
\textit{Dijkstra}'s \cite{dijkstra1959note} is the fundamental shortest path algorithm that uses the best-first strategy to find the shortest paths from source to all the other vertices in the distance-increasing order. \textit{Sky-Dijk}\footnote{It was called \textit{MinSum-MinSum} problem as \textit{skyline} was not used in 1980s.} \cite{hansen1980bicriterion} is the first algorithm for skyline path query that expands the search space incrementally like the \textit{Dijkstra}'s but each vertex might be visited multiple times with one for each skyline result. Dominance relation can be applied to prune the search space. \cite{kriegel2010route} further incorporates \textit{landmark} \cite{kriegel2007proximity} to estimate the lower bounds, and prunes the search space with dominance relations. \cite{gong2019skyline} utilizes the skyline operator to find a set of skyline destinations.

\subsection{Constrained Shortest Path}
\label{subsec:RelatedWork_CSP}
As the single constraint shortest path \textit{CSP} is the basis of \textit{MCSP}, we discuss the four kinds of \textit{CSP} algorithms depending on the \textit{exact/approximate} and \textit{index-free/index} classification.

\paragraph{\textbf{Exact Index-Free CSP}.}
\cite{joksch1966shortest} is the first work studies \textit{CSP} and solves it with dynamic programming. \cite{handler1980dual} converts it to an \textit{integer linear programming (ILP)} problem and solves it with a standard \textit{ILP} solution. However, both of them are not scalable to large real-life road networks. Among the practical solutions, the \textit{Skyline Dijkstra} algorithm can be utilized to answer \textit{CSP} queries by setting a constraint on the cost during the skyline search. 
Another practical solution to answer \textit{CSP} is through \textit{k-Path} \cite{yen1971finding}, which tests the top-$k$ paths one by one until satisfying $C$. \textit{Pulse} \cite{LOZANO2013378} improves this procedure by applying constraint pruning, but its \textit{DFS} search strategy limits its performance. \textit{cKSP} \cite{gao2010fast} uses the pruned \textit{Dijkstra} search in its $k$-path generation and incorporates the skyline dominance relation to prune the infeasible ones. \textit{eKSP} \cite{SEDENONODA2015602} further replaces \textit{cKSP}'s search with sub-path concatenation to speed up $k$-path generation. 

\paragraph{\textbf{Approximate Index-Free CSP}.}
The approximation allows the result length being not longer than $(1+\alpha)$ times of the optimal path. The first approximate solution is also proposed by \cite{hansen1980bicriterion} with a complexity of $O(|E|^2$\mbox{\scriptsize $\dfrac{|V|^2}{\alpha-1}$} $\log$ \mbox{\scriptsize $\dfrac{|V|}{\alpha-1}$}), and \cite{lorenz2001simple} improves it to $O(|V||E|$ $(\log\log$\mbox{\scriptsize $\dfrac{w_{max}}{w_{min}}+\dfrac{1}{\alpha-1}))$}. \textit{Lagrange Relaxation} is applied by \cite{handler1980dual,juttner2001lagrange,carlyle2008lagrangian} to obtain the approximation result with $O(|E|\log^3|E|)$ iterations of path finding. As shown in \cite{kuipers2006comparison} and \cite{LOZANO2013378}, these approximate algorithms are even orders of magnitude slower than the \textit{k-Path}-based exact algorithms. \textit{CP-CSP} \cite{tsaggouris2009multiobjective} approximates the \textit{Skyline Dijkstra} by distributing the approximation power $\sqrt[n]{\alpha}$ to the edges. However, this approximation ratio would decrease to nearly 1 when the graph is big, so its performance is only slightly better than \textit{Sky-Dijk}.

\paragraph{\textbf{Exact Index-Based CSP}.}
There are many shortest path indexes, but only \textit{CH} \cite{geisberger2008contraction} has been extended to \textit{CSP} \cite{storandt2012route}, 
with its shortcuts being \textit{skyline paths}. Consequently, its construction and query answering are both \textit{Sky-Dijk}-based and suffers from long index construction, large index size, and slow query answering. Hence, \textit{CSP-CH} only keeps a limited number of skyline paths in each shortcut to reduce the construction time, but its query performance suffers.

\paragraph{\textbf{Approximate Index-Based CSP}.}
The only approximate index method is \textit{COLA} \cite{wang2016effective}, which partitions the graph into regions and precomputes approximate skyline paths between regions. Similar to \textit{CP-CSP}, the approximate paths are also computed in the \textit{Sky-Dijk} fashion. In order to avoid the diminishing approximation power $\sqrt[n]{|V|}$, it views $\alpha$ as a budget and concentrates it on the important vertices. Similar pruning technique is also applied in the time-dependent path finding \cite{li2019Time}. In this way, the approximation power is released so it has faster construction time and smaller index size.

\subsection{Multi-Constraint Shortest Path}
Most of the existing \textit{MCSP} algorithms have very high complexity and cannot scale to real-life networks. For instance, \cite{jaffe1984algorithms} has a very high complexity of $O(|V|^5c_{max}$ $\log(|V|c_{max}))$. \cite{korkmaz2001multi} applies a heuristic-based approximation algorithm with no approximation bound, a backward \textit{Dijkstra}'s for feasibility check and a forward \textit{Dijkstra}'s to connect the results. However, it can only scale to small graph with 200 vertices. \cite{shi2017multi} takes seconds on the graphs with only 100 vertices. \cite{kriegel2010route,yang2021efficient} restricts the skyline path to the pre-defined ``preference function", so they cannot be utilized to answer \textit{MCSP} problem. \cite{DBLP:conf/adc/ZhangLZL21} compares and conducts experiments on the skyline-based and $k$-shortest path-based methods. 

\subsection{2-Hop Labeling}
\label{subsec:RelatedWork_2Hop}
The 2-hop labeling index \cite{cohen2003reachability} answers distance query without any search. \textit{PLL} \cite{akiba2013fast,zhang2021DWPSL} and \textit{PSL} \cite{li2019scaling} are designed for \textit{small-world} networks since several large degree vertices work as gateways to help reduce the search. \textit{Tree-Decomposition}-based methods \textit{H2H} \cite{ouyang2018hierarchy} and \textit{TEDI} \cite{wei2012tedi} are suitable for road networks because of the low average degree and the small tree width. \textit{CT-Index} \cite{li2020scaling} combines \textit{H2H} and \textit{PLL} to deal with the periphery and core of a graph separately.

\section{Conclusion}
\label{sec:Conclusion}
In this paper, we have proposed forest-based skyline path index structure to answer the \textit{MCSP} query efficiently. Moreover, we have proposed several skyline path concatenation methods to support index construction and query answering. The new path concatenation paradigm enables indexing and storing the complete exact skyline paths. A range of pruning techniques have been developed for the single pair concatenation, multiple hop concatenation, \textit{MCSP} query answering, and forest query answering. Extensive evaluations on real-life road networks have demonstrated that our new methods can outperform the state-of-the-art \textit{MCSP} approaches by several orders of magnitude, and our pruning techniques can provide a further 10x speedup.


\begin{acknowledgements}
This work is partially supported by the Australian Research Council (grants DP200103650 and LP180100018)..
\end{acknowledgements}

%
%
\bibliographystyle{abbrv}
\bibliography{document}

\balance
\end{document}